%% file: main.tex
\documentclass[a4paper]{llncs}

\pdfoutput=1

\newcommand{\defref}[1]{\autoref{#1}}

\newcommand{\lemref}[1]{\autoref{#1}}

\usepackage{inputenc}
\usepackage{comment}
\usepackage{paralist}

\usepackage{wrapfig}
\usepackage{microtype}

\usepackage{amsfonts,amsmath,amssymb}
\usepackage{color}
\usepackage{stmaryrd}
\usepackage{tikz}
\usetikzlibrary{positioning}
\usepackage{forloop}
\usepackage{multirow}

\usepackage{listings-rust,paralist}
\usepackage{url}
\usepackage{hyperref}

\usepackage[ruled, vlined]{algorithm2e}
\SetKwProg{Fn}{Function}{}{end}
\SetAlgoLongEnd

\graphicspath{{.}{figures/}}

\usepackage{etoolbox}
\BeforeBeginEnvironment{tabular}{\vspace*{-2mm}}
\AfterEndEnvironment{tabular}{\vspace*{-2mm}}

\usepackage{xspace}
\usepackage{float}
\usepackage{xcolor}
 \definecolor{chameleond}{HTML}{4E9A06}
 \definecolor{skyblued}{HTML}{204A87}
 \definecolor{bluue}{HTML}{0047AB}

\input{notations}

\title{Analysing Parallel Complexity of Term Rewriting\thanks{
  This work was partially funded by the French National Agency of Research in
  the CODAS Project (ANR-17-CE23-0004-01).}}

\author{Thaïs Baudon\inst{1} \and Carsten Fuhs\inst{2} \and Laure Gonnord\inst{3,1}}

\institute{LIP (UMR CNRS/ENS Lyon/UCB Lyon1/INRIA), Lyon, France
 \and
Birkbeck, University of London, United Kingdom
 \and
LCIS (UGA/Grenoble INP/Ésisar), Valence, France}
    
\begin{document}
\setlength\abovedisplayshortskip{2pt}
\setlength\belowdisplayshortskip{2pt}
\setlength\abovedisplayskip{3pt}
\setlength\belowdisplayskip{3pt}
\maketitle

\begin{abstract}
  We revisit
  parallel-innermost term rewriting as a model of parallel computation
  on inductive data structures and provide a corresponding notion of
  runtime complexity parametric in the size of the start term.  We
  propose automatic techniques to derive both upper and lower bounds
  on parallel complexity of rewriting that enable a direct reuse of
  existing techniques for sequential complexity.
  The applicability and the precision of the method are demonstrated by the relatively light effort in extending the program analysis tool \aprove and by experiments on numerous benchmarks from the literature.

\end{abstract}

\section{Introduction}
\label{sec:intro}
Automated inference of complexity bounds for parallel computation has
seen a surge of attention in recent years \cite{BaillotPiESOP21,BaillotPiCONCUR21,GallagherLOPSTR19,AlbertTOCL18,HoffmannESOP15,HoffmannICFP18}. While
techniques and tools for a variety of computational models have been
introduced, so far there does not seem to be any paper in this area for
complexity of \emph{term rewriting} with
parallel
evaluation
strategies.
This paper addresses this gap in the
\begin{wrapfigure}[7]{r}{7.4cm}
\vspace*{-6.7ex}
\begin{lstlisting}[language=Rust,style=colouredRust]
fn size(&self) -> int {
 match self {
  &Tree::Node { v, ref left, ref right }
   => left.size() + right.size() + 1,
  &Tree::Empty => 0 , }    } 
\end{lstlisting}
\vspace*{-2ex}
\caption[Size of Tree in Rust]{Tree size computation in Rust}
\label{lst:rustsize}
\end{wrapfigure}
literature. We consider
term rewrite systems (TRSs) as \emph{intermediate representation} for
programs with \emph{pattern-matching} operating on \emph{algebraic data
  types}
like the one depicted in \autoref{lst:rustsize}.

In this particular example, the recursive calls to
\lstinline[language=Rust,style=colouredRust]{left.size()} and
\lstinline[language=Rust,style=colouredRust]{right.size()} can be done
in parallel.
Building on previous work on parallel-innermost
rewriting~\cite{parallelRewriting,innermostOrdering}, and first ideas about parallel
complexity~\cite{wst16trs}, we propose a new notion of Parallel
Dependency Tuples that captures such a behaviour, and methods to
compute both upper and lower \emph{parallel complexity bounds}. 

Bounds on parallel complexity can provide insights about the
potentiality of parallelisation: if sequential and parallel complexity
of a function
(asymptotically) coincide, this information can be useful for a
parallelising compiler to refrain from parallelising the evaluation of
this function.
Moreover, evaluation of TRSs (as a simple
functional programming language) in massively
parallel settings such as GPUs is currently a topic of active research
\cite{gpu_trs}. In this context, a static analysis of parallel complexity can be
helpful to determine whether
to rewrite on a
(fast, but not very parallel) CPU or on a (slower, but massively parallel) GPU.

A preliminary version of this work with an initial notion of parallel
complexity was presented in
an informal
extended abstract
\cite{wst21trs}.
We now propose a
more formal version accompanied by extensions, proofs,
implementation, experiments, and related work.
\autoref{sec:irc} recalls term rewriting
and Dependency Tuples \cite{DependencyTuple}
as the basis of our approach.
In \autoref{sec:para_complex}, we
introduce a notion of runtime complexity for parallel-innermost
rewriting, and we harness the existing Dependency Tuple framework
to compute asymptotic upper bounds on this complexity.
In \autoref{sec:dt_to_irc}, we provide a transformation to innermost
term rewriting that lets any tool for
(sequential) innermost runtime complexity be reused
to find upper bounds for parallel-innermost runtime complexity and,
for
confluent parallel-innermost rewriting, also
lower bounds.
\autoref{sec:expe} gives experimental
evidence of the practicality of our method on a large standard benchmark set.
We discuss related work
in~\autoref{sec:related}.
\confOrReport{Our extended authors' accepted manuscript \cite{versionArxiv}
of this paper
under Creative Commons CC-BY licence additionally
has full proofs of our theorems.}{This technical report provides
full proofs in the appendix. It is an
extended authors' accepted manuscript
for our LOPSTR 2022 paper \cite{lopstr2022}.}

\section{Term Rewriting and Innermost Runtime Complexity}
\label{sec:irc}

We assume basic familiarity with term rewriting (see, e.g.,
\cite{BaaderNipkow}) and recall standard definitions to fix notation.
As customary for analysis of runtime complexity of rewriting,
we consider terms as \emph{tree-shaped} objects,
without sharing of subtrees.

We first define \emph{Term Rewrite Systems} and \emph{Innermost Rewriting}.
$\TT(\Signature, \VV)$ denotes the set of \emph{terms} over a
finite signature $\Signature$ and the set of variables $\VV$.
For a term $t$, its
\emph{size} $\tsize{t}$
is defined by:
(a)
if $t \in \VV$,
$\tsize{t} = 1$;
(b)
if $t = f(t_1, \ldots, t_n)$, then
$\tsize{t} = 1 + \sum_{i = 1}^n \tsize{t_i}$.
The set $\Pos(t)$ of
the \emph{positions} of $t$ is
defined by:
(a) if $t \in \VV$, then $\Pos(t) = \{ \varepsilon \}$,
and (b) if $t = f(t_1,\ldots,t_n)$, then
$\Pos(t) = \{ \varepsilon \} \cup
 \bigcup_{1 \leq i \leq n}\{ i \pi \mid \pi \in \Pos(t_i) \}$.
The position $\varepsilon$ is
the \emph{root position}
of term $t$.
If $t = \asym(t_1,\ldots,t_n)$,
$\rt(t) = \asym$ is the \emph{root symbol} of $t$.
The \emph{(strict) prefix order} $>$ on positions is
the strict partial order given by:
$\tau > \pi$ iff there exists $\pi' \neq \varepsilon$ such that $\pi
\pi' = \tau$.
Two positions $\pi$ and $\tau$ are \emph{parallel} iff neither
$\pi > \tau$ nor $\pi = \tau$ nor $\tau > \pi$ hold.
For $\pi \in \Pos(t)$, $t|_\pi$ is the subterm of $t$
at position $\pi$, and we write $t[s]_\pi$ for the term that results
from $t$ by replacing the subterm $t|_\pi$ at position $\pi$ by the term $s$.

A substitution $\sigma$ is a mapping from $\VV$ to
$\TT(\Signature, \VV)$ with finite domain
$\Dom(\sigma) = \{ x \in \VV \mid \sigma(x) \neq x \}$.
We write
$\{x_1 \mapsto t_1; \ldots; x_n \mapsto t_n\}$ for
a substitution $\sigma$ with $\sigma(x_i) = t_i$ for
$1 \leq i \leq n$ and $\sigma(x) = x$ for all
other $x \in \VV$.
We extend substitutions to terms by
$\sigma(f(t_1,\ldots,f_n)) = f(\sigma(t_1),\ldots,\sigma(t_n))$.
We may write $t\sigma$ for $\sigma(t)$.

For a term $t$, $\VV(t)$ is the set of variables
in
$t$.
A \emph{term rewrite system (TRS)} $\RR$ is a set of rules
   $\{ \ell_1 \to r_1,  \ldots, \ell_n \to r_n \}$
  with
    $\ell_i, r_i \in \TT(\Signature, \VV)$,
    $\ell_i \not\in \VV$,
  and $\VV(r_i) \subseteq \VV(\ell_i)$ for all $1 \leq i \leq n$.
The \emph{rewrite relation} of $\RR$ is
$s \to_\RR t$ iff
  there are
   a rule $\ell \to r \in \RR$,
   a position $\pi \in \Pos(s)$,
   and a substitution $\sigma$
  such that
   $s = s[\ell\sigma]_\pi$ and
   $t = s[r\sigma]_\pi$.
  Here, $\sigma$ is called the \emph{matcher} and the term $\ell\sigma$
  the \emph{redex} of the rewrite step.
  If  no proper subterm of  $\ell\sigma$ is
  a possible redex,
  $\ell\sigma$ is an \emph{innermost redex}, and the rewrite step
   is an \emph{innermost rewrite step}, denoted by $s \ito t$.

 $\DefSyms^{\RR}=\{ f \mid f(\ell_1,\ldots,\ell_n) \to r \in \RR \}$ and
 $\ConSyms^{\RR}= \Signature \setminus \DefSyms^{\RR}$
 are the \emph{defined}  and \emph{constructor} symbols of $\RR$.
 We may
 also just write $\DefSyms$ and $\ConSyms$.
 The set of positions with defined symbols of $t$ is $\PosDef(t) = \{ \pi \mid \pi \in \Pos(t), \rt(t|_\pi) \in \DefSyms \}$.

 For a relation $\to$, $\to^+$ is its transitive closure
 and $\to^*$ its reflexive-transitive closure. An object $o$ is a \emph{normal
 form} wrt a relation $\to$ iff there is no $o'$ with $o \to o'$.
 A relation $\to$ is \emph{confluent} iff $s \to^* t$ and
 $s \to^* u$ implies that
 there exists an object $v$ with
 $t \to^* v$ and $u \to^* v$.
 A relation $\to$ is \emph{terminating} iff
 there is no infinite sequence $t_0 \to t_1 \to t_2 \to \cdots$.

\begin{example}[$\Fsize$]
\label{ex:size1}
Consider the TRS $\RR$ with the following rules modelling the code of~\autoref{lst:rustsize}.
\[
\begin{array}{rcl@{\hspace*{5ex}}|@{\hspace*{5ex}}rcl}
\Fplus(\FZero, y) &\to& y & \Fsize(\FNil) &\to &\FZero \\
\Fplus(\FS(x), y) &\to& \FS(\Fplus(x, y)) &
\Fsize(\FTree(v, l, r)) &\to& \FS(\Fplus(\Fsize(l), \Fsize(r)))
\end{array}
\]
Here $\DefSyms^{\RR} = \{ \Fplus, \Fsize \}$ and
$\ConSyms^{\RR} = \{ \FZero, \FS, \FNil, \FTree \}$.
We have the following innermost rewrite sequence, where
the used innermost redexes are underlined:
\[
\begin{array}{rl}
&\underline{\Fsize(\FTree(\FZero, \FNil, \FTree(\FZero, \FNil, \FNil)))}\\
\ito & \FS(\Fplus(\underline{\Fsize(\FNil)}, \Fsize(\FTree(\FZero, \FNil, \FNil))))\\
\ito & \FS(\Fplus(\FZero, \underline{\Fsize(\FTree(\FZero, \FNil, \FNil))}))\\
\ito & \FS(\Fplus(\FZero, \FS(\Fplus(\underline{\Fsize(\FNil)}, \Fsize(\FNil)))))\\
\ito & \FS(\Fplus(\FZero, \FS(\Fplus(\FZero, \underline{\Fsize(\FNil)}))))\\
\ito & \FS(\Fplus(\FZero, \FS(\underline{\Fplus(\FZero, \FZero)})))\\
\ito & \FS(\underline{\Fplus(\FZero, \FS(\FZero))})\\
\ito & \FS(\FS(\FZero))
\end{array}
\]
This rewrite sequence uses 7 steps to reach a normal form.

\end{example}

We wish to provide static bounds on the
length of the longest
rewrite sequence from terms of a specific size.
Here we
use innermost evaluation strategies,
which closely correspond to call-by-value strategies used in many
programming languages.
We focus
on rewrite sequences that start with \emph{basic terms},
corresponding to function calls where a function is applied to data
objects. The
resulting
notion of complexity for term
rewriting is known as \emph{innermost runtime complexity}.

\begin{definition}[Innermost Runtime Complexity $\irc{}$ \cite{Hirokawa08IJCAR,DependencyTuple}]
\label{def:rc}
The \emph{derivation height} of a term $t$ wrt\ a relation $\to$ is the
length of the longest sequence of $\to$-steps
from $t$:
$\Dh(t, \to) = \sup \{ e \mid \exists\, t' \in \TT(\Signature, \VV).\;
t \to^e t' \}$ where $\to^e$ is the $e^{\textrm{th}}$
iterate of $\to$.
If $t$ starts an infinite $\to$-sequence, we write $\Dh(t, \to) =
\omega$.
Here, $\omega$ is the smallest infinite ordinal, i.e.,
$\omega > n$ holds for all $n \in \nat$.

A term $f(t_1, \ldots, t_k)$ is \emph{basic (for a TRS $\RR$)}
iff $f\in\DefSyms^{\RR}$ and $t_1, \dots, t_k \in \TT(\ConSyms^{\RR}, \VV)$.
$\BasicTerms^{\RR}$ is the set of basic terms
for a TRS $\RR$.
For $n \in \mathbb{N}$,
  the \emph{innermost runtime complexity} function
  is
  $\ircR(n) = \sup \{ \Dh(t, {\ito}) \mid t \in \BasicTerms^{\RR}, \tsize{t} \leq
n \}$.
For all $P \subseteq \nat \cup \{\omega \}$, $\sup\, P$ is the
least upper bound of $P$, where $\sup\, \emptyset = 0$.

\end{definition}

Many automated techniques are available
\cite{Hirokawa08IJCAR,DependencyTuple,HirokawaMoser14,ava:mos:16,naa:fro:bro:fuh:gie:17,MoserS20}
to analyse $\ircR$.  We build on Dependency Tuples
\cite{DependencyTuple}, originally designed to find upper bounds
for\linebreak (sequential) innermost runtime complexity.  A central
idea is to group all function calls by a rewrite rule \emph{together}
rather than to separate them (as with DPs for proving termination
\cite{DependencyPairs}).  We use \emph{sharp terms} to represent these
function calls.

\begin{definition}[Sharp Terms $\SharpTerms$]
For every $f \in \DefSyms$, we introduce a fresh symbol
$\tup{f}$ of the same arity, called a \emph{sharp symbol}.
For a term $t = f(t_1,\ldots,t_n)$ with
$f \in \DefSyms$, we define $\tup{t} = \tup{f}(t_1,\ldots,t_n)$.
For all other terms $t$, we define $\tup{t} = t$.
$\SharpTerms = \{ \tup{t} \mid t \in \TT(\Signature, \VV),
\rt(t) \in \DefSyms \}$ denotes the set of \emph{sharp terms}.
\end{definition}

To
get an upper bound for sequential complexity, we
``count'' how often
each rewrite rule is used.  The idea is
that when a rule $\ell \to r$ is used,
the cost (i.e., number of rewrite steps for the evaluation)
of the function call to the instance of $\ell$
is 1 + the sum
of the costs of all the function calls in the resulting instance of $r$,
counted separately in some fixed order.
To group $k$ function calls together, we use ``compound symbols''
$\FCom_k$ of arity $k$, which intuitively represent the sum of
the runtimes of their arguments.

\begin{definition}[Dependency Tuple, DT \cite{DependencyTuple}]
\label{def:dt}
A \emph{dependency tuple (DT)} is a rule of the form
$\tup{s} \to \FCom_n(\tup{t}_1,\ldots,\tup{t}_n)$
where $\tup{s}, \tup{t}_1,\ldots,\tup{t}_n \in \SharpTerms$.
Let $\ell \to r$ be a rule with $\PosDef(r) = \{ \pi_1, \ldots, \pi_n \}$
and $\pi_1 \gtrdot \ldots \gtrdot \pi_n$ for a total order $\gtrdot$
(e.g., lexicographic order) on positions.
Then $\DT(\ell \to r) = \tup{\ell} \to \FCom_n(\tup{r|}_{\pi_1},\ldots,\tup{r|}_{\pi_n})$.\footnote{The order $\gtrdot$ must be total to ensure that the function
$\DT$ is well defined wrt the order of the arguments of $\FCom_n$. The
(partial!)\ prefix order $>$ is not sufficient here.
 }
For a TRS $\RR$, let $\DT(\RR) = \{ \DT(\ell \to r) \mid \ell \to r \in \RR \}$.
\end{definition}

\begin{example}
\label{ex:sizeDTs}
For $\RR$ from \autoref{ex:size1}, $\DT(\RR)$ consists of the
following DTs:
\[
\begin{array}{rcl}
\FPLUS(\FZero, y) & \to &\FCom_0 \\
\FPLUS(\FS(x), y) & \to &\FCom_1(\FPLUS(x, y)) \\
\FSIZE(\FNil) & \to &\FCom_0 \\
\FSIZE(\FTree(v, l, r)) &\to&
   \FCom_3(\FSIZE(l), \FSIZE(r), \FPLUS(\Fsize(l), \Fsize(r)))
\end{array}
\]

To represent the complexity of a sharp term for a set of DTs and
a TRS $\RR$,
\emph{chain trees}
are used \cite{DependencyTuple}.
Intuitively, a chain tree for some sharp term is a dependency tree
of the computations involved in evaluating this term.
Each node represents a computation (the DT) on some arguments
(defined by the substitution).

\begin{definition}[Chain Tree, $\Cplx$ \cite{DependencyTuple}]
Let $\DD$ be a set of DTs and $\RR$ be a TRS.
Let $T$ be a (possibly infinite) tree where each node is labelled with
a DT $\tup{q} \to \FCom_n(\tup{w}_1,\ldots,\tup{w}_n)$ from $\DD$ and
a substitution $\nu$, written
$(\tup{q} \to \FCom_n(\tup{w}_1,\ldots,\tup{w}_n) \mid \nu)$.
Let the root node
be labelled with $(\tup{s} \to \FCom_e(\tup{r}_1,\ldots,\tup{r}_e) \mid \sigma)$.
Then
$T$ is a \emph{$(\DD,\RR)$-chain tree for $\tup{s} \sigma$} iff
the following conditions hold for any node of $T$,
where $(\tup{u} \to \FCom_m(\tup{v}_1,\ldots,\tup{v}_m) \mid \mu)$
is the label of the node:
\begin{itemize}
\item $\tup{u}\mu$ is in normal form wrt~$\RR$;
\item if this node has the children 
$(\tup{p}_1 \to \FCom_{m_1}(\ldots) \mid \delta_1), \ldots,
(\tup{p}_k \to \FCom_{m_k}(\ldots) \mid \delta_k)$, then there are
pairwise different $i_1,\ldots,i_k \in \{1,\ldots,m\}$ with
$\tup{v}_{i_j} \mu \itos \tup{p}_j \delta_j$ for all $j \in \{1,\ldots,k\}$.
\end{itemize}

Let $\SSS \subseteq \DD$
and $\tup{s} \in \SharpTerms$. For a chain tree $T$,
$|T|_\SSS \in \mathbb{N} \cup \{\omega\}$
is the number of nodes in
$T$
labelled with a DT from $\SSS$. We define
$\Cplx_{\langle \DD, \SSS, \RR \rangle}(\tup{s}) = \sup \{ |T|_\SSS \mid
T \text{ is a } (\DD,\RR)\text{-chain tree for }
\tup{s} \}$.
For terms $\tup{s}$ without a $(\DD,\RR)$-chain tree,
we define $\Cplx_{\langle \DD, \SSS, \RR \rangle}(\tup{s}) = 0$.
\end{definition}
\begin{example}
    For $\RR$ from \autoref{ex:size1} and $\DD = \DT(\RR)$ from
    \autoref{ex:sizeDTs}, the following is a
    chain tree for the term
    $\FSIZE(\FTree(\FZero, \FNil, \FNil))$:

    \noindent
    \input{chain_tree.tikz}
\end{example}

The main correctness statement in the sequential case is the
following:

\begin{theorem}[$\Cplx$ bounds Derivation Height for $\ito$
\cite{DependencyTuple}]
\label{thm:cplxSeq}
Let $\RR$ be a TRS, let $t = f(t_1,\ldots,t_n) \in \TT(\Signature,\VV)$
such that all $t_i$ are in normal form
(this includes all $t \in \BasicTerms^{\RR}$). Then we have $\Dh(t,\ito) \leq
\Cplx_{\langle \DT(\RR), \DT(\RR), \RR \rangle}(\tup{t})$.
If $\ito$ is confluent, then  $\Dh(t,\ito) =
\Cplx_{\langle \DT(\RR), \DT(\RR), \RR \rangle}(\tup{t})$.
\end{theorem}

For automated complexity analysis with DTs, the following notion of
\emph{DT problems} is used as a characterisation of DTs
that we reduce in incremental proof steps
to a trivially solved problem.

\begin{definition}[DT Problem, Complexity of DT Problem \cite{DependencyTuple}]
\label{def:dt_problem}
Let $\RR$ be a TRS, $\DD$ be a set of DTs, $\SSS \subseteq \DD$.
Then $\langle \DD, \SSS, \RR \rangle$ is a DT problem.
Its complexity function is
$\irc{\langle \DD, \SSS, \RR \rangle}(n) =
\sup \{ \Cplx_{\langle \DD, \SSS, \RR \rangle}(\tup{t}) \mid
        t \in \BasicTerms^\RR, |t| \leq n \}$.
The DT problem
$\langle \DT(\RR), \DT(\RR), \RR \rangle$ is called the
\emph{canonical DT problem} for $\RR$.
\end{definition}

For a DT problem $\langle \DD, \SSS, \RR \rangle$, the set $\DD$
contains all DTs that can be used in chain trees.
$\SSS$ contains the DTs whose complexity remains to be
analysed. $\RR$ contains the rewrite rules for evaluating
the arguments of DTs. Here we focus on simplifying $\SSS$ (thus $\DD$ and $\RR$ are fixed during the process)
but techniques to simplify $\DD$ and $\RR$ are available as well
\cite{DependencyTuple,ava:mos:16}.

\autoref{thm:cplxSeq} implies the following link between
$\ircR$ and $\irc{\langle \DT(\RR), \DT(\RR), \RR \rangle}$:

\begin{theorem}[Complexity Bounds for TRSs via Canonical DT
Problems \cite{DependencyTuple}]
Let $\RR$ be a TRS with canonical DT problem
$\langle \DT(\RR), \DT(\RR), \RR \rangle$.
Then we have $\ircR(n) \leq \irc{\langle \DT(\RR), \DT(\RR), \RR \rangle}(n)$.
If $\ito$ is confluent, we have
$\ircR(n) = \irc{\langle \DT(\RR), \DT(\RR), \RR \rangle}(n)$.
\end{theorem}

In practice, the focus
is on
finding
asymptotic bounds  for $\ircR$.
For example, \autoref{ex:size:irc} will show that
for our TRS $\RR$
from \autoref{ex:size1}
we have $\ircR(n) \in \OO(n^2)$.

A DT problem $\langle \DD, \SSS, \RR \rangle$ is said to be \emph{solved}
iff $\SSS = \emptyset$: we always have $\irc{\langle \DD, \emptyset, \RR \rangle}(n) = 0$.
To simplify and finally solve DT problems in an incremental fashion,
complexity analysis techniques called \emph{DT processors} are used. A
DT processor takes a DT problem as input and returns a (hopefully
simpler) DT problem as well as an asymptotic complexity bound as an
output. The largest asymptotic complexity bound returned over this
incremental process is then also an upper bound for $\ircR(n)$
\cite[Corollary 21]{DependencyTuple}.

The
reduction pair processor using polynomial interpretations
\cite{DependencyTuple}
applies a restriction
of polynomial interpretations to $\nat$ \cite{Lankford75} to infer
upper bounds on the number of times that DTs can occur
in a chain tree for terms of size at most $n$.

\begin{definition}[Polynomial Interpretation, CPI]
A \emph{polynomial interpretation} $\Pol$ maps
every $n$-ary function symbol to a polynomial with variables
$x_1,\ldots,x_n$ and coefficients from $\nat$.
$\Pol$ extends
to terms via $\Pol(x)=x$ for
$x \in \VV$ and $\Pol(f(t_1,\ldots,t_n)) =
\Pol(f)(\Pol(t_1),\ldots,\Pol(t_n))$.
$\Pol$ induces an order $\succ_\Pol$\linebreak and a quasi-order
$\succsim_\Pol$ over terms where
$s \succ_\Pol t$ iff $\Pol(s) > \Pol(t)$
and
$s \succsim_\Pol t$ iff $\Pol(s) \geq \Pol(t)$
for all instantiations of variables with natural numbers.

A \emph{complexity polynomial interpretation (CPI)} $\Pol$
is a polynomial interpretation where:
$\Pol(\FCom_n(x_1,\ldots,x_n)) = x_1 + \dots + x_n$, and
for all
$f \in \ConSyms$,
$\Pol(f(x_1,\ldots,x_n)) = a_1\cdot x_1 + \dots + a_n \cdot x_n + b$
for some $a_i \in \{0,1\}$ and $b \in \nat$.
\end{definition}

The restriction for CPIs regarding constructor symbols
enforces that the interpretation of a constructor term $t$ (as an argument
of a term for which a chain tree is constructed) can
exceed its size $\tsize{t}$ only by at most a constant factor.
This is crucial for soundness.
Using a CPI, we can now define and state correctness of the
corresponding reduction pair processor
\cite[Theorem 27]{DependencyTuple}.

\begin{theorem}[Reduction Pair Processor with CPIs
  \cite{DependencyTuple}]
\label{thm:redpair}
Let $\langle \DD, \SSS, \RR \rangle$ be a DT problem,
let $\succsim$ and $\succ$ be induced by a CPI $\Pol$.
Let $k \in \nat$ be the maximal degree of all polynomials
$\Pol(\tup{f})$ for all $f \in \DefSyms$.
Let $\DD \cup \RR \subseteq {\succsim}$.
If $\SSS \cap {\succ} \neq \emptyset$,
the reduction pair processor returns the DT problem
$\langle \DD, \SSS \setminus{\succ}, \RR \rangle$
and the complexity $\OO(n^k)$.
Then the reduction pair processor is sound.
\end{theorem}

\begin{example}[\autoref{ex:sizeDTs} continued]
\label{ex:size:irc}
For our running example,
consider the CPI
$\Pol$
with:
$\Pol(\FPLUS(x_1,x_2)) = \Pol(\Fsize(x_1)) = x_1,\;
\Pol(\FSIZE(x_1)) = 2 x_1 + x_1^2,
\linebreak
\Pol(\Fplus(x_1,x_2)) = x_1 + x_2,
\Pol(\FTree(x_1,x_2,x_3)) = 1 + x_2 + x_3,
\Pol(\FS(x_1)) = 1 + x_1,
\Pol(\FZero) = \Pol(\FNil) = 1$.
$\Pol$ orients all DTs in
$\SSS = \DT(\RR)$ with $\succ$ and all
rules in $\RR$ with $\succsim$.
This proves
$\ircR(n) \in \OO(n^2)$:
since the maximal degree of the CPI for a symbol $\tup{f}$ is 2,
the upper bound of $\OO(n^2)$ follows by
\autoref{thm:redpair}.
\end{example}

\end{example}

\section{Finding Upper Bounds for Parallel Complexity}
\label{sec:para_complex}

In this section we present our main contribution: an application of the
DT framework from innermost runtime complexity to \emph{parallel-innermost rewriting}.

The notion of parallel-innermost rewriting dates back at least
to \cite{parallelRewriting}. Informally, in a parallel-innermost
rewrite step, all innermost redexes are rewritten simultaneously.
This corresponds to executing all function calls in parallel using a
call-by-value strategy on a machine with unbounded parallelism \cite{BlellochFPCA95}.
In the literature \cite{LoopsUnderStrategies}, this strategy is also known as
``max-parallel-innermost rewriting''.

\begin{definition}[Parallel-Innermost Rewriting \cite{innermostOrdering}]
\label{def:pito}
A term $s$ \emph{rewrites innermost in parallel} to $t$ with a TRS $\RR$,
written $s \pito t$, iff $s \itop t$, and either
(a) $s \ito t$ with $s$ an innermost redex, or
(b) $s = f(s_1, \ldots, s_n)$, $t = f(t_1, \ldots, t_n)$, and for all
$1 \leq k \leq n$ either $s_k \pito t_k$ or $s_k = t_k$ is a normal form.
\end{definition}

\begin{example}[\autoref{ex:size1} continued]
\label{ex:pito_sequence}
The TRS $\RR$ from \autoref{ex:size1}
allows the following parallel-innermost rewrite sequence, where
innermost redexes are underlined:
\[
\begin{array}{rl}
&\underline{\Fsize(\FTree(\FZero, \FNil, \FTree(\FZero, \FNil, \FNil)))}\\
\pito & \FS(\Fplus(\underline{\Fsize(\FNil)}, \underline{\Fsize(\FTree(\FZero, \FNil, \FNil))}))\\
\pito & \FS(\Fplus(\FZero, \FS(\Fplus(\underline{\Fsize(\FNil)}, \underline{\Fsize(\FNil)}))))\\
\pito & \FS(\Fplus(\FZero, \FS(\underline{\Fplus(\FZero, \FZero)})))\\
\pito & \FS(\underline{\Fplus(\FZero, \FS(\FZero))})\\
\pito & \FS(\FS(\FZero))
\end{array}
\]

In the second and in the third step, two innermost steps each
happen in parallel (which is not possible with standard
innermost rewriting: $\mathrel{\pito}\ \not\subseteq\ \mathrel{\ito}$).
An innermost rewrite sequence
without parallel evaluation
necessarily needs
two more steps to a normal form from this start term,
as in \autoref{ex:size1}.
\end{example}

Note that for all TRSs $\RR$, $\pito$ is terminating
iff $\ito$ is terminating \cite{innermostOrdering}.
\autoref{ex:pito_sequence}
\linebreak
shows that
such an equivalence does \emph{not} hold for the derivation height
of a term.
The question now is: given a TRS $\RR$, how much of a speed-up might
we get by
a switch
from innermost to parallel-innermost rewriting?
To investigate,
we extend the notion of innermost
runtime complexity to parallel-innermost rewriting.

\begin{definition}[Parallel-Innermost Runtime Complexity $\pirc{}$]
\label{def:pirc}
For $n \in \mathbb{N}$, we define the
  \emph{parallel-innermost runtime complexity} function as
  $\pircR(n) = \sup \{ \Dh(t, {\pito}) \mid t \in \BasicTerms^{\RR}, \tsize{t} \leq
n \}$.
\end{definition}

In the literature on parallel computing
\cite{BlellochFPCA95,HoffmannESOP15,BaillotPiESOP21},
the terms \emph{depth} or
\emph{span} are commonly used for the concept of the runtime
of a function on a machine with unbounded parallelism (``wall time''), corresponding
to the complexity measure of $\pircR$.
In contrast, $\ircR$ would describe the \emph{work} of a function
(``CPU time'').

In the following, given a TRS $\RR$, our goal shall be to infer
(asymptotic) upper bounds for $\pircR$ fully automatically.
Of course, an upper bound for (sequential) $\ircR$
is also an upper bound for $\pircR$.
We will now introduce
techniques to find upper bounds for $\pircR$
that are strictly tighter than these trivial bounds.

To find upper bounds for runtime complexity of parallel-innermost
rewriting, we can \emph{reuse} the notion of DTs
from \autoref{def:dt} for sequential innermost rewriting along with
existing techniques~\cite{DependencyTuple} as illustrated in the following example.

\vspace{-1mm}
\begin{example}
\label{ex:sizePDTs}
In the recursive $\Fsize$-rule, the two calls to $\Fsize(l)$ and
$\Fsize(r)$ happen \emph{in parallel} (they are \emph{structurally
independent})
and take place at \emph{parallel positions} in the term.
Thus, the cost (number of rewrite steps with $\pito$ until a normal form is reached) for these two
calls is not the \emph{sum}, but the \emph{maximum} of their
individual costs.
Regardless of which of these two calls has the higher cost,
we still need to add the cost for the call to $\Fplus$
on the results of the two calls:
$\Fplus$ starts evaluating only after both calls to $\Fsize$ have
finished.
With $\sigma$ as the used matcher for the rule
and with $t \downarrow$ as the (here unique)
normal form resulting from repeatedly rewriting a term $t$
with $\pito$ (the ``result'' of evaluating $t$), we have:
\[
\begin{array}{rl}
&\Dh(\Fsize(\FTree(v, l, r)) \sigma, \pito)\\
=&
1 + \max( \Dh(\Fsize(l) \sigma, \pito), \Dh(\Fsize(r) \sigma, \pito))\\
& \hspace*{38ex}
 {} + \Dh(\Fplus(\Fsize(l) \sigma \!\downarrow, \Fsize(r) \sigma \!\downarrow), \pito)
\end{array}
\]

In the DT setting, we could
introduce a new symbol $\FComPar_k$
that explicitly expresses that its arguments
are evaluated in parallel. This symbol would
then be interpreted as the maximum
of its arguments in an extension of \autoref{thm:redpair}:
\begin{align*}
\FSIZE(\FTree(v, l, r)) & \to
   \FCom_2(\FComPar_2(\FSIZE(l), \FSIZE(r)), \FPLUS(\Fsize(l), \Fsize(r)))
\end{align*}
Although automation of the search for
polynomial interpretations extended by the maximum function is
readily available~\cite{maxpolo},
we would still have to extend the notion of
Dependency Tuples and also adapt all existing techniques
in the Dependency Tuple framework to work with $\FComPar_k$.

This is why we have chosen the following alternative approach,
which is equally powerful on theoretical level and
enables immediate reuse of existing techniques in the DT framework.
Equivalently to the above, we can ``factor in'' the cost of
calling $\Fplus$ into the maximum function:
\[
\begin{array}{rl}
&\Dh(\Fsize(\FTree(v, l, r)) \sigma, \pito)\\
=&
\max( 1 + \Dh(\Fsize(l) \sigma, \pito) + \Dh(\Fplus(\Fsize(l) \sigma \!\downarrow, \Fsize(r) \sigma \!\downarrow), \pito),\\
&\hspace*{5.2ex} 1 + \Dh(\Fsize(r) \sigma, \pito) +
\Dh(\Fplus(\Fsize(l) \sigma \!\downarrow, \Fsize(r) \sigma \!\downarrow), \pito))
\end{array}
\]
Intuitively, this would correspond to evaluating
$\Fplus(\ldots, \ldots)$
twice, in two parallel threads of execution,
which costs the same amount of (wall) time as evaluating
$\Fplus(\ldots, \ldots)$
once.
We can represent this maximum of the execution times of two threads
by introducing \emph{two} DTs for our recursive $\Fsize$-rule:
\[
\begin{array}{rcl}
\FSIZE(\FTree(v, l, r)) & \to &
   \FCom_2(\FSIZE(l), \FPLUS(\Fsize(l), \Fsize(r)))\\
\FSIZE(\FTree(v, l, r)) & \to &
   \FCom_2(\FSIZE(r), \FPLUS(\Fsize(l), \Fsize(r)))
\end{array}
\]
To express the cost of a concrete rewrite sequence, we would
non-deterministically choose the DT that
corresponds to the ``slower thread''.
\end{example}

In other words,
when a rule $\ell \to r$ is used,
the cost of the function call to the instance of $\ell$
is 1 + the sum
of the costs of the function calls in the resulting instance of $r$ \emph{that are in structural dependency with each other}.
The actual cost of the function call to the instance of $\ell$ in a concrete
rewrite sequence is the \emph{maximum} of all the possible costs
caused by such \emph{chains} of structural dependency
(based on the prefix order $>$ on positions of defined
function symbols in $r$).
Thus, \emph{structurally independent} function calls are considered in
separate DTs, whose non-determinism models the parallelism
of these function calls.

The notion of \emph{structural dependency} of function calls is captured by
\autoref{def:structdeps}. Basically, it comes from the fact that a term
cannot be evaluated before all its subterms have been reduced to normal forms
(innermost rewriting/\emph{call by value}).
This induces a ``happens-before'' relation for the computation
\cite{Lamport78}.
\begin{definition}[Structural Dependency, $\msdcSym$]
  \label{def:structdeps}
For positions $\pi_1, \ldots, \pi_k$, we call $\lmsdc \pi_1, \dots, \pi_k \rmsdc$
a \emph{structural dependency chain} for a term $t$ iff
$\pi_1, \ldots, \pi_k \in \PosDef(t)$ and $\pi_1 > \ldots > \pi_k$.
Here $\pi_i$ \emph{structurally depends on} $\pi_j$ in $t$ iff $j < i$.
A structural dependency chain $\lmsdc \pi_1, \dots, \pi_k \rmsdc$
for a term $t$ is \emph{maximal} iff
$k = 0$ and $\PosDef(t) = \emptyset$, or $k > 0$ and
$\forall \pi \in \PosDef(t)\, .\, \pi \ngtr \pi_1 \wedge (\pi_1 > \pi
\Rightarrow \pi \in \{\pi_2, \ldots, \pi_k\})$.
We write $\msdcSym(t)$ for the set of all maximal structural
dependency chains for $t$.
\end{definition}

Note that $\msdcSym(t) \neq \emptyset$ always holds:
if $\PosDef(t) = \emptyset$, then
$\msdcSym(t) = \{ \lmsdc \rmsdc \}$.

\begin{example}
Let $t = \FS(\Fplus(\Fsize(\FNil), \Fplus(\Fsize(x),\FZero)))$.
In our running example,
$t$ has the
 following
structural dependencies:
$\msdcSym(t) = \{ \lmsdc 11, 1 \rmsdc, \lmsdc 121, 12, 1 \rmsdc \}$.
The
chain $\lmsdc 11, 1 \rmsdc$
corresponds to the nesting of
$t|_{11}=\Fsize(\FNil)$ below $t|_{1}=\Fplus(\Fsize(\FNil),
\Fplus(\Fsize(x),\FZero))$, so the evaluation of
$t|_{1}$ will have to wait at least until $t|_{11}$ has been
fully evaluated.

If
$\pi$ structurally depends on $\tau$ in a
term $t$,
neither $t|_\tau$ nor $t|_\pi$ need to \emph{be} a redex.\linebreak
Rather, $t|_\tau$ could be \emph{instantiated} to a
redex and
an instance of $t|_\pi$ could become a redex after
its subterms, including the instance of $t|_\tau$, have been
evaluated.
\end{example}

We thus revisit the notion of DTs, which now embed structural
dependencies in addition to the algorithmic dependencies
already captured in DTs.

\begin{definition}[Parallel Dependency Tuples $\DTpar$, 
Canonical Parallel DT Problem]
For a rewrite rule $\ell \to r$, we define the set of its
\emph{Parallel Dependency Tuples (PDTs)} $\DTpar(\ell \to r)$:
$\DTpar(\ell \to r) = \{ \tup{\ell} \to \FCom_k(\tup{r|}_{\pi_1},\ldots,\tup{r|}_{\pi_k})
 \mid
 \msdc{r}{\pi_1,\ldots, \pi_k} \}$.
For a TRS $\RR$, let $\DTpar(\RR) = \bigcup_{\ell \to r \in \RR}
\DTpar(\ell \to r)$.
The \emph{canonical parallel DT problem} for $\RR$ is
$\langle \DTpar(\RR), \DTpar(\RR), \RR \rangle$.
\end{definition}

\begin{example}
For our recursive $\Fsize$-rule
$\ell \to r$,
we have
$\PosDef(r) = \{ 1, 11, 12 \}$
and
$\msdcSym(r) = \{ \lmsdc 11, 1 \rmsdc, \lmsdc 12, 1 \rmsdc \}$.
With $\mathit{r}|_1 = \Fplus(\Fsize(l), \Fsize(r))$,
$\mathit{r}|_{11} = \Fsize(l)$, and $\mathit{r}|_{12} = \Fsize(r)$,
we get the PDTs
from \autoref{ex:sizePDTs}. For
the rule $\Fsize(\FNil) \to \FZero$,
we have $\msdcSym(\FZero) = \{ \lmsdc \rmsdc \}$,
so we get
$\DTpar(\Fsize(\FNil) \to \FZero) = \{ \FSIZE(\FNil) \to \FCom_0 \}$.
\end{example}

We can now make our main correctness statement:

\begin{theorem}[$\Cplx$ bounds Derivation Height for $\pito$]
\label{thm:cplxPar}
Let $\RR$ be a TRS, let $t = f(t_1,\ldots,t_n) \in \TT(\Signature,\VV)$
such that all $t_i$ are in normal form
(e.g., when $t \in \BasicTerms^{\RR}$). Then we have $\Dh(t,\pito) \leq
\Cplx_{\langle \DTpar(\RR), \DTpar(\RR), \RR \rangle}(\tup{t})$.
If $\pito$ is confluent, then  $\Dh(t,\pito) =
\Cplx_{\langle \DTpar(\RR), \DTpar(\RR), \RR \rangle}(\tup{t})$.\footnote{
The proof uses the confluence of $\RR$  as a sufficient criterion for \emph{unique normal forms}.}
\end{theorem}

From \autoref{thm:cplxPar}, the soundness of our approach to parallel
complexity analysis via the DT framework follows analogously to
\cite{DependencyTuple}:

\begin{theorem}[Parallel Complexity Bounds for TRSs via Canonical Parallel DT
Problems]
\label{thm:canonical_pdt_problem}
Let $\RR$ be a TRS with canonical parallel DT problem
$\langle \DTpar(\RR), \DTpar(\RR), \RR \rangle$.
Then we have $\pircR(n) \leq \irc{\langle \DTpar(\RR), \DTpar(\RR), \RR \rangle}(n)$.
If $\pito$ is confluent, we have
$\pircR(n) = \irc{\langle \DTpar(\RR), \DTpar(\RR), \RR \rangle}(n)$.
\end{theorem}

This theorem implies that we can reuse arbitrary techniques to find
upper bounds
for \emph{sequential} complexity in the
DT framework also to find upper bounds for \emph{parallel}
complexity, without requiring any modification to the framework.

Thus, via \autoref{thm:redpair},
in particular we can use polynomial interpretations in the
DT framework for our PDTs to get upper bounds for $\pircR$.

\begin{example}[\autoref{ex:sizePDTs} continued]
\label{ex:size_tight_bound}
For our TRS $\RR$ computing the $\Fsize$ function on trees, we get the
set $\DTpar(\RR)$ with the following
PDTs:
\[
\begin{array}{rcl}
\FPLUS(\FZero, y) & \to &  \FCom_0 \\
\FPLUS(\FS(x), y) & \to & \FCom_1(\FPLUS(x, y)) \\
\FSIZE(\FNil) & \to & \FCom_0 \\
\FSIZE(\FTree(v, l, r)) & \to &
   \FCom_2(\FSIZE(l), \FPLUS(\Fsize(l), \Fsize(r))) \\
\FSIZE(\FTree(v, l, r)) & \to &
   \FCom_2(\FSIZE(r), \FPLUS(\Fsize(l), \Fsize(r)))
\end{array}
\]
The
interpretation $\Pol$
from \autoref{ex:size:irc} implies
$\pircR(n) \in \OO(n^2)$.
This bound is tight: consider $\Fsize(t)$ for a comb-shaped tree $t$
where the first argument of $\FTree$ is always $\FZero$ and the
third is always $\FNil$.
The function $\Fplus$, which needs time linear in its first argument,
is called linearly often on data linear in the size of the start term.
Due to the structural dependencies, these calls do not happen in parallel
(so call $k+1$ to $\Fplus$ must wait for call $k$).
\end{example}

\begin{example}
\label{ex:doubles}
Note that $\pircR(n)$ can be asymptotically lower than $\ircR(n)$, for instance for the TRS $\RR$ with the following rules:
\[
\begin{array}{rcl@{\hspace*{5ex}}|@{\hspace*{5ex}}rcl}
\Fdoubles(\FZero) & \to & \FNil &
  \Fd(\FZero) & \to & \FZero \\
\Fdoubles(\FS(x)) & \to & \FCons(\Fd(\FS(x)), \Fdoubles(x)) &
  \Fd(\FS(x)) & \to & \FS(\FS(\Fd(x)))
\end{array}
\]
The upper bound $\ircR(n) \in \mathcal{O}(n^2)$ is tight:
from
$\Fdoubles(\FS(\FS(\ldots\FS(\FZero)\ldots)))$, we get
linearly many calls to the linear-time function $\Fd$ on arguments
of size linear in the start term.
However, the Parallel Dependency Tuples in this example are:
\[
\begin{array}{rcl@{\hspace*{5ex}}|@{\hspace*{5ex}}rcl}
\FDOUBLES(\FZero) & \to & \FCom_0 & \FD(\FZero) & \to & \FCom_0 \\
\FDOUBLES(\FS(x)) & \to & \FCom_1(\FD(\FS(x))) & \FD(\FS(x)) & \to & \FCom_1(\FD(x)) \\
\FDOUBLES(\FS(x)) & \to & \FCom_1(\FDOUBLES(x))
\end{array}
\]
Then the following polynomial 
interpretation, which orients all DTs with $\succ$ and all
rules from $\RR$ with $\succsim$, proves $\pircR(n) \in \mathcal{O}(n)$:
$\Pol(\FDOUBLES(x_1)) = \Pol(\Fd(x_1)) = 2 x_1,
\Pol(\FD(x_1)) = x_1,
\Pol(\Fdoubles(x_1)) = \Pol(\FCons(x_1,x_2)) = \Pol(\FZero) = \Pol(\FNil) = 1,
\Pol(\FS(x_1)) = 1 + x_1$.
\end{example}

Interestingly enough, Parallel Dependency Tuples also allow us to
identify TRSs that have \emph{no} potential for parallelisation
by parallel-innermost rewriting.

\begin{theorem}[Absence of Parallelism by PDTs]
\label{thm:nopar}
Let $\RR$ be a TRS
such that for all rules $\ell \to r \in \RR$,
$|\msdcSym(r)| = 1$.
Then:
(a) $\DTpar(\RR) = \DT(\RR)$;
(b) for all basic terms $t_0$ and rewrite sequences
$t_0 \pito t_1 \pito t_2 \pito \dots$, also
$t_0 \ito  t_1 \ito  t_2 \ito  \dots$ holds (i.e., from basic terms,
$\pito$ and $\ito$ coincide);
(c) $\pirc{\RR}(n) = \irc{\RR}(n)$.
\end{theorem}

Thus, for TRSs $\RR$ where
\autoref{thm:nopar} applies,
no rewrite rule
can introduce
parallel redexes, and
specific analysis techniques for
$\pircR$
are not needed.

\section{From Parallel DTs to Innermost Rewriting}
\label{sec:dt_to_irc}
As we have seen in the previous section, we can transform a TRS $\RR$ with
parallel-innermost rewrite relation to a DT problem whose
complexity provides an upper bound of $\pircR$ (or,
for
confluent $\pito$, corresponds exactly to $\pircR$).
However,
DTs are only one of many
available techniques
to find bounds for $\ircR$. Other techniques
include, e.g.,
Weak Dependency Pairs \cite{Hirokawa08IJCAR},
usable replacement maps \cite{HirokawaMoser14},
the Combination Framework \cite{ava:mos:16},
a transformation to complexity
problems for integer transition systems \cite{naa:fro:bro:fuh:gie:17},
amortised complexity analysis \cite{MoserS20},
or techniques for finding \emph{lower} bounds \cite{LowerBounds}.
Thus, can we benefit also from other
techniques for (sequential) innermost complexity to analyse parallel
complexity?

In this section, we answer the question in the affirmative, via a
generic transformation from Dependency Tuple problems back to rewrite
systems whose innermost complexity can then be analysed using
arbitrary existing techniques.

We use
\emph{relative rewriting},
which allows for labelling some of the rewrite rules such that
their use does not contribute to the derivation height of a term.
In other words, rewrite steps with these rewrite rules are
``for free'' from the perspective of complexity.
Existing state-of-the-art tools like \aprove~\cite{aprove-tool}
and \tct~\cite{tct} are
able to find bounds on (innermost) runtime complexity of such
rewrite systems.

\begin{definition}[Relative Rewriting]
\label{def:rel}
For two TRSs $\RRA$ and $\RRB$,  $\RRA/\RRB$ is a
 \emph{relative TRS}.
Its \emph{rewrite relation} $\to_{\RRA/\RRB}$ is
  $\to^*_{\RRB} \circ \to_{\RRA} \circ \to^*_{\RRB}$, i.e.,  rewriting with
$\RRB$ is allowed before and after each $\RRA$-step.
We define the \emph{innermost} rewrite relation by
$s \someito{\RRA/\RRB} t$ iff
  $s \to^*_{\RRB} s' \to_{\RRA} s'' \to^*_{\RRB} t$  for some terms $s',
  s''$ such that the
proper subterms of the
redexes of each step with $\to_\RRB$ or $\to_\RRA$ are in normal form
wrt\ $\RRA \cup \RRB$.

The set $\BasicTerms^{\RRA/\RRB}$ of basic terms for a
relative TRS $\RRA/\RRB$ is
$\BasicTerms^{\RRA/\RRB} = \BasicTerms^{\RRA\cup\RRB}$.
The notion of innermost runtime complexity extends to
relative TRSs in the natural way:
$\irc{\RRA/\RRB}(n) =
\sup \{ \Dh(t, {\someito{\RRA/\RRB}}) \mid t \in \BasicTerms^{\RRA/\RRB}, \tsize{t}
\leq n \}$
\end{definition}

The rewrite relation $\someito{\RRA/\RRB}$ is essentially
the same as $\someito{\RRA \cup \RRB}$, but only steps using
rules from $\RRA$ count towards the complexity;
steps
using rules from $\RRB$ have no cost.
This can be useful, e.g., for representing that built-in functions
from programming languages modelled as recursive functions
have constant cost.

\begin{example}
Consider a variant of \autoref{ex:size1} where
$\Fplus(\FS(x), y) \to \FS(\Fplus(x, y))$ is moved
to $\RRB$,
but all other rules are elements of $\RRA$.
Then
$\RRA/\RRB$ would provide a modelling of the $\Fsize$ function that is
closer to the Rust function from \autoref{sec:intro}.
Let $\FS^n(\FZero)$ denote the term obtained by $n$-fold application
of $\FS$ to $\FZero$ (e.g., $\FS^2(\FZero) = \FS(\FS(\FZero))$).
Although $\Dh(\Fplus(\FS^n(\FZero),\FS^m(\FZero)), \someito{\RRA\cup\RRB}) =
n+1$, we would then get $\Dh(\Fplus(\FS^n(\FZero),\FS^m(\FZero)),
\someito{\RRA/\RRB}) = 1$, corresponding to a machine model where the
time of evaluating addition for integers is constant.
\end{example}

Note the similarity of a relative TRS and a Dependency Tuple
problem: only certain rewrite steps count towards the
analysed complexity.
We make use of this observation for the following transformation.

\begin{definition}[Relative TRS for a Dependency Tuple Problem, $\detup$]
\label{def:to_relative}
Let $\langle \DD, \SSS, \RR \rangle$ be a Dependency Tuple problem.
We define the corresponding relative TRS
$\detup(\langle \DD, \SSS, \RR \rangle) =
\SSS/((\DD \setminus \SSS) \cup \RR)$.
\end{definition}

In other words, we omit the information that steps with our
dependency tuples can happen only on top level (possibly below
constructors $\FCom_n$, but above $\someto{\RR}$ steps). (As we shall see in \autoref{thm:lower_rel}, this information can be recovered.)

The following example is taken from the
\emph{Termination Problem Data Base (TPDB)} \cite{tpdb},
a collection of examples used at the
annual \emph{Termination and Complexity Competition (termCOMP)}
\cite{termcomp,termcompWiki} (see also \autoref{sec:expe}):

\begin{example}[TPDB, \texttt{HirokawaMiddeldorp\_04/t002}]
\label{ex:to_relative}
Consider the following TRS $\RR$ from category
\texttt{Innermost\_Runtime\_Complexity}
of the TPDB:
\[
\begin{array}{rcl@{\hspace*{5ex}}|@{\hspace*{5ex}}rcl}
\Fleq(\Fzero, y) & \to & \Ftrue
& \Fif(\Ftrue, x, y) & \to & x\\
\Fleq(\Fs(x), \Fzero) & \to & \Ffalse
& \Fif(\Ffalse, x, y) & \to & y\\
\Fleq(\Fs(x), \Fs(y)) & \to & \Fleq(x, y)
& \Fminus(x, \Fzero) & \to & x\\
\Fmod(\Fzero, y) & \to & \Fzero
& \Fminus(\Fs(x), \Fs(y)) & \to & \Fminus(x, y)\\
\Fmod(\Fs(x), \Fzero) & \to & \Fzero\\
\Fmod(\Fs(x), \Fs(y)) & \to & \multicolumn{4}{l}{\Fif(\Fleq(y, x), \Fmod(\Fminus(\Fs(x), \Fs(y)), \Fs(y)), \Fs(x))
}
\end{array}
\]
This TRS has the following PDTs $\DTpar(\RR)$:
\[
\begin{array}{rcl@{\hspace*{5ex}}|@{\hspace*{5ex}}rcl}
\tup{\Fleq}(\Fzero, y) & \to & \FCom_0
& \tup{\Fif}(\Ftrue, x, y) & \to & \FCom_0\\
\tup{\Fleq}(\Fs(x), \Fzero) & \to & \FCom_0
& \tup{\Fif}(\Ffalse, x, y) & \to & \FCom_0\\
\tup{\Fleq}(\Fs(x), \Fs(y)) & \to & \FCom_1(\tup{\Fleq}(x, y))
& \tup{\Fminus}(x, \Fzero) & \to & \FCom_0\\
\tup{\Fmod}(\Fzero, y) & \to & \FCom_0
& \tup{\Fminus}(\Fs(x), \Fs(y)) & \to & \FCom_1(\tup{\Fminus}(x, y))\\
\tup{\Fmod}(\Fs(x), \Fzero) & \to & \FCom_0\\
\tup{\Fmod}(\Fs(x), \Fs(y)) & \to &
                                    \multicolumn{4}{l}{\FCom_2(\tup{\Fleq}(y,x), \tup{\Fif}(\Fleq(y, x), \Fmod(\Fminus(\Fs(x), \Fs(y)), \Fs(y)), \Fs(x)))}\\
\tup{\Fmod}(\Fs(x), \Fs(y)) & \to &
                                    \multicolumn{4}{l}{\FCom_3(\tup{\Fminus}(\Fs(x), \Fs(y)), \tup{\Fmod}(\Fminus(\Fs(x), \Fs(y)), \Fs(y)),}\\
&& \multicolumn{4}{r}{\tup{\Fif}(\Fleq(y, x), \Fmod(\Fminus(\Fs(x), \Fs(y)), \Fs(y)), \Fs(x)))}
\end{array}
\]

The canonical parallel DT problem is
$\langle \DTpar(\RR), \DTpar(\RR), \RR \rangle$.
We get the relative TRS
$\detup(\langle \DTpar(\RR), \DTpar(\RR), \RR \rangle) = \DTpar(\RR)/\RR$.

\end{example}

\begin{theorem}[Upper Complexity Bounds for $\detup(\langle \DD, \SSS,
\RR \rangle)$ from
\label{thm:upper_rel}
  $\langle \DD, \SSS, \RR \rangle$]
Let $\langle \DD, \SSS, \RR \rangle$ be a
DT problem.
Then (a) for all
$\tup{t} \in \SharpTerms$ with $t \in \BasicTerms^\RR$,
we have $\Cplx_{\hspace*{-1pt}\langle \DD, \SSS, \RR \rangle}(\tup{t})
\,{\leq} \Dh(\tup{t}, \itodetup)$,
and (b) $\irc{\langle \DD, \SSS, \RR \rangle}(n) \,{\leq}
\,\irc{\detupTRS}(n)$.
\end{theorem}

\begin{example}[\autoref{ex:to_relative} continued]
\label{ex:to_relative_upper}
For the relative TRS $\DTpar(\RR)/\RR$
from
\autoref{ex:to_relative},
the
tool \aprove\ uses a transformation to
integer transition systems~\cite{naa:fro:bro:fuh:gie:17} followed by an application of the
complexity analysis tool \cofloco~\cite{CoFloCo,CoFloCoFM16}
to find
a bound $\irc{\DTpar(\RR)/\RR}(n) \in \OO(n)$ and
to deduce the bound $\pirc{\RR}(n) \in \OO(n)$ for the original TRS $\RR$
from the TPDB. In contrast, using the techniques of~\autoref{sec:para_complex}
without the transformation to a relative
TRS from \autoref{def:to_relative}, \aprove\ finds only a bound $\pirc{\RR}(n) \in \OO(n^2)$.
\end{example}

Intriguingly, we can use our transformation
from \autoref{def:to_relative} not only for
finding upper bounds, but also for \emph{lower} bounds on $\pirc{\RR}$.

\begin{theorem}[Lower Complexity Bounds for $\detup(\langle \DD, \SSS, \RR \rangle)$ from
  $\langle \DD, \SSS, \RR \rangle$]
\label{thm:lower_rel}
Let $\langle \DD, \SSS, \RR \rangle$ be a
DT
problem.
Then (a) there is a type assignment
s.t.\ for all\linebreak
$\ell \to r \in \DD \cup \RR$, $\ell$ and $r$ get the same type, and
for all well-typed $t \in \BasicTerms^{\DD \cup \RR}$,
$\Cplx_{\hspace*{-0.5ex}\langle \DD, \SSS, \RR \rangle}(\tup{t}) \geq
\Dh(t, \itodetup)$, and
(b) $\irc{\hspace*{-1pt}\langle \DD, \SSS, \RR \rangle}(n) \geq
\irc{\detupTRS}(n)$.
\end{theorem}

\autoref{thm:upper_rel}
and \autoref{thm:lower_rel}
hold
regardless of
whether the original
DT problem was obtained from a TRS with sequential or with parallel
evaluation.
So while this
kind of connection between DT (or DP) problems and
relative rewriting may be
folklore in the community, its application
to convert a TRS whose \emph{parallel} complexity is sought
to a TRS with the same \emph{sequential} complexity is new.
\smallskip

Note that \autoref{thm:canonical_pdt_problem} requires confluence of
$\pito$ to derive lower bounds for $\pirc{\RR}$
from lower complexity bounds of the canonical parallel DT problem.
So
to use \autoref{thm:lower_rel} to search for \emph{lower}
complexity bounds with existing techniques~\cite{LowerBounds},
we need a criterion for confluence of parallel-innermost rewriting.

\begin{example}[Confluence of $\ito$ does not Imply Confluence of
$\pito$]
To see that we cannot prove confluence of $\pito$ just by
using a standard off-the-shelf tool for
confluence analysis of innermost or full rewriting \cite{coco}, consider the TRS
$\RR = \{\Fa \to \Ff(\Fb,\Fb), \Fa \to \Ff(\Fb,\Fc), \Fb \to \Fc,
\Fc \to \Fb\}$. For this TRS, both $\ito$ and $\someto{\RR}$ are confluent.
However, $\pito$ is not confluent: we can rewrite both
$\Fa \pito \Ff(\Fb,\Fb)$ and $\Fa \pito \Ff(\Fb,\Fc)$,
yet there
is no term $v$ such that $\Ff(\Fb,\Fb) \pitos v$ and
$\Ff(\Fb,\Fc) \pitos v$. The reason is that the only possible rewrite
sequences with $\pito$ from these terms are
$\Ff(\Fb,\Fb) \pito \Ff(\Fc,\Fc) \pito \Ff(\Fb,\Fb) \pito \dots$
and
$\Ff(\Fb,\Fc) \pito \Ff(\Fc,\Fb) \pito \Ff(\Fb,\Fc) \pito \dots$,
with no terms in common.
\end{example}

\begin{conjecture}
If $\pito$ is confluent, then $\ito$ is confluent.
\end{conjecture}

Confluence means:
if a term $s$ can be rewritten to two
different terms $t_1$ and $t_2$ in 0 or more steps,
it is
always possible to rewrite $t_1$ and $t_2$ in 0 or more steps
to
a
term $u$.
For
$\pito$, the redexes that get rewritten are
fixed: all
innermost redexes simultaneously.
Thus,
$s$ can
rewrite to two \emph{different} terms $t_1$ and $t_2$
only if at least one of these redexes can be rewritten in two
different ways using
$\ito$.

Towards a sufficient criterion for confluence of parallel-innermost
rewriting, we introduce the following standard definition:

\begin{definition}[Non-Overlapping]
A TRS $\RR$ is \emph{non-overlapping} iff for any two rules
$\ell \to r, u \to v \in \RR$ where variables have been renamed apart
between the rules, there is no position $\pi$ in $\ell$ such that
$\ell|_\pi \notin \VV$
and the terms $\ell|_\pi$ and $u$ unify.
\end{definition}

A sufficient criterion that a given redex has a unique result from a
rewrite step is given in the following.

\begin{lemma}[\cite{BaaderNipkow}, Lemma 6.3.9]
If a TRS $\RR$ is non-overlapping,
$s \someto{\RR} t_1$ and
$s \someto{\RR} t_2$ with the redex of both rewrite steps at the same
position, then $t_1 = t_2$.
\end{lemma}

With the above reasoning, this lemma directly gives us a sufficient criterion for confluence of
\emph{parallel-innermost} rewriting.

\begin{corollary}[Confluence of Parallel-Innermost Rewriting]
\label{cor:confluence}
If a TRS $\RR$ is non-overlapping, then $\pito$ is confluent.
\end{corollary}

So, in those cases we can actually use this sequence of
transformations from a parallel-innermost TRS via a DT problem to an
innermost (relative) TRS to analyse both upper and lower bounds for
the original. Conveniently, these cases correspond to
deterministic
programs, our motivation for this work!

\begin{example}[\autoref{ex:to_relative_upper} continued]
\autoref{cor:confluence} and
\autoref{thm:lower_rel} imply that a lower bound for
$\irc{\DTpar(\RR)/\RR}(n)$ of the relative TRS $\DTpar(\RR)/\RR$
from \autoref{ex:to_relative} carries over to
$\pirc{\RR}(n)$ of the original TRS $\RR$ from the
TPDB.
\aprove uses rewrite lemmas \cite{LowerBounds} to find
the lower bound $\irc{\DTpar(\RR)/\RR}(n) \in \Omega(n)$.
Together with \autoref{ex:to_relative_upper}, we have
automatically inferred
that this complexity bound is \emph{tight}:
$\pirc{\RR}(n) \in \Theta(n)$.
\end{example}

\section{Implementation and Experiments}
\label{sec:expe}

We have implemented the contributions of this paper in the
automated termination and complexity analysis tool
\aprove~\cite{aprove-tool}. We added or modified 620 lines of Java code,
including
\begin{inparaenum}
\item the framework of parallel-innermost rewriting;
\item the generation of parallel DTs (\autoref{thm:canonical_pdt_problem});
\item a processor to convert them to TRSs with the same complexity
(\autoref{thm:upper_rel}, \autoref{thm:lower_rel});
\item the confluence test of~\autoref{cor:confluence}.
\end{inparaenum}
As far as we are aware, this is the first implementation of a fully
automated inference of complexity bounds for parallel-innermost
rewriting.
To demonstrate the effectiveness of our implementation, we have 
considered the 663 TRSs from category
\texttt{Runtime\_Complexity\_Innermost\_Rewriting} of the
TPDB, version 11.2~\cite{tpdb}.
This category
of the TPDB
is the benchmark collection used at
termCOMP to compare tools that
infer complexity bounds for runtime complexity of innermost
rewriting, $\ircR$.
To get meaningful results, we first applied \autoref{thm:nopar} to
exclude TRSs $\RR$ where $\pircR(n) = \ircR(n)$ trivially holds.
We obtained
294 TRSs with potential for parallelism as our benchmark set.
We conducted our experiments
on the \starexec\ compute cluster \cite{starexec} in
the \texttt{all.q} queue. The timeout per example and tool
configuration was set to 300 seconds.
Our experimental data with analysis times
and all examples are available online~\cite{evalPage}.

As remarked earlier, we always have $\pircR(n) \leq \ircR(n)$,
so an upper bound for $\ircR(n)$ is always a legitimate
upper bound for $\pircR(n)$.
Thus,
we include
upper bounds for $\ircR$
found by the state-of-the-art tools \aprove\ and \tct~\cite{versionPage,tct}.
from
termCOMP 2021
as a ``baseline'' in our evaluation.
We compare with several configurations of \aprove\ and \tct\ that use
the
techniques of this paper for $\pircR$:
``\aprove\ $\pircR$ Section~3'' also uses
\autoref{thm:canonical_pdt_problem} to produce canonical
parallel DT problems as input for the DT framework.
``\aprove\ $\pircR$ Sections~3~\&~4'' additionally uses
the transformation from \autoref{def:to_relative}
to convert a TRS $\RR$ to a relative TRS
$\DTpar(\RR)/\RR$
and then to analyse $\irc{\DTpar(\RR)/\RR}(n)$ (for lower bounds
only together with a confluence proof
via \autoref{cor:confluence}).
We also extracted each of the TRSs $\DTpar(\RR)/\RR$
and used the files as inputs for
\aprove\ and \tct\ from termCOMP 2021.
``\aprove\ $\pircR$ Section~4''
and ``\tct\ $\pircR$ Section~4'' provide the results for
$\irc{\DTpar(\RR)/\RR}$
(for lower bounds, only where $\pito$ had been proved confluent).

\begin{table}[t]
\begin{center}
\begin{tabular}{|l||c|c|c|c|c|}
\hline
Tool & $\Oh(1)$ & $\leq\Oh(n)$ & $\leq\Oh(n^2)$ & $\leq\Oh(n^3)$
     & $\leq\Oh(n^{\geq 4})$  \\\hline\hline
\tct\ $\ircR$
& 4 & 28 & 39 & 44 & 44 \\ %
\aprove\ $\ircR$
& \textbf{5} & 50 & 110 & 123 & 127 \\ %
\aprove\ $\pircR$ Section \ref{sec:para_complex}
& \textbf{5} & 65 & \textbf{125} & \textbf{140} & \textbf{142} \\ %
\aprove\ $\pircR$ Sections \ref{sec:para_complex} \&
  \ref{sec:dt_to_irc}
& \textbf{5} & \textbf{69} & \textbf{125} & 139 & 141 \\ %
\hline
\tct\ $\pircR$ Section \ref{sec:dt_to_irc}
& 3 & 39 & 52 & 56 & 57 \\
\aprove\ $\pircR$ Section \ref{sec:dt_to_irc}
& \textbf{5} & 62 & 96 & 105 & 105 \\
\hline
\end{tabular}
\vspace{3mm}
\caption{Upper bounds for runtime complexity of (parallel-)innermost
rewriting}
\label{table:ui}
\begin{tabular}{|l||c||c|c|c|c|}
\hline
Tool
& confluent
& $\geq\Omega(n)$ & $\geq\Omega(n^2)$ & $\geq\Omega(n^3)$
  & $\geq\Omega(n^{\geq 4})$ \\\hline\hline
\aprove\ $\pircR$ Sections \ref{sec:para_complex} \&
  \ref{sec:dt_to_irc}
& \textbf{186}
& 133 & \textbf{23} & \textbf{5} & \textbf{1} \\ %
\hline
\tct\ $\pircR$ Section \ref{sec:dt_to_irc}
& \textbf{186}
& 59 & 0 & 0 & 0 \\
\aprove\ $\pircR$ Section \ref{sec:dt_to_irc}
& \textbf{186}
& \textbf{155} & 22 & \textbf{5} & \textbf{1} \\
\hline
\end{tabular}
\vspace{3mm}
\caption{Lower bounds for runtime complexity of parallel-innermost
rewriting}
\label{table:li}
\begin{tabular}{|l||c|c|c|c||c|}
\hline
Tool & $\Theta(1)$ & $\Theta(n)$ & $\Theta(n^2)$ & $\Theta(n^3)$ & Total \\\hline\hline
\aprove\ $\pircR$ Sections \ref{sec:para_complex} \&
  \ref{sec:dt_to_irc}
& \textbf{5} & 32 & \textbf{1} & \textbf{3} & 41 \\ %
\hline
\tct\ $\pircR$ Section \ref{sec:dt_to_irc}
& 3 & 21 & 0 & 0 & 24 \\
\aprove\ $\pircR$ Section \ref{sec:dt_to_irc}
& \textbf{5} & \textbf{37} & \textbf{1} & \textbf{3} & \textbf{46} \\
\hline
\end{tabular}
\vspace{3mm}
\caption{Tight bounds for runtime complexity of parallel-innermost
rewriting}
\label{table:tight}
\end{center}
\vspace*{-7mm}
\end{table}

\autoref{table:ui}
gives an overview over our experimental results for upper bounds.
For each configuration, we state the number of examples for
which the corresponding asymptotic complexity bound
was
inferred.
A
column ``$\leq \Oh(n^k)$'' means that the corresponding tools proved a bound
$\leq \Oh(n^k)$ (e.g., the configuration
``\aprove\ $\ircR$'' proved constant or linear upper bounds in
50
cases). Maximum values in a column are highlighted in bold.
We observe that upper complexity bounds improve in a
noticeable number of cases, e.g., linear bounds on $\pircR$ can now be inferred
for
69 TRSs rather than for
50 TRSs (using upper bounds on $\ircR$
as an over-approximation), an improvement by 38\%.
Note that this does \emph{not} indicate deficiencies in the existing tools for $\ircR$,
which had not been designed with analysis of $\pircR$ in mind -- rather,
it shows that specialised techniques for analysing $\pircR$ are a worthwhile
subject of investigation.
Note also that \autoref{ex:size:irc} and
\autoref{ex:size_tight_bound} show that even for TRSs
with potential for parallelism, the actual parallel and sequential
complexity may still be asymptotically identical,
which further highlights the need for dedicated analysis techniques for $\pircR$.

The improvement from $\ircR$ to $\pircR$ can be drastic:
for example, for the TRS \texttt{TCT\_12/recursion\_10}, the
bounds found by \aprove\ change from an upper bound of sequential
complexity of $\OO(n^{10})$ to a (tight) upper bound for parallel
complexity of $\OO(n)$. (This TRS models a specific recursion
structure, with rules
$\{ \Ff_0(x) \to \Fa \} \cup
 \{ \Ff_i(x) \to \Fg_i(x,x), \; \Fg_i(\Fs(x), y) \to \Fb(\Ff_{i-1}(y),
 \Fg_i(x,y)) \mid 1 \leq i \leq 10 \}$, and is highly amenable to parallelisation.)
We observe that adding the
techniques from \autoref{sec:dt_to_irc} to the techniques from
\autoref{sec:para_complex} leads to only few examples for which better
upper bounds can be found (one of them is \autoref{ex:to_relative_upper}).

\autoref{table:li} shows our results for lower bounds on $\pircR$.
Here we evaluated only
configurations including
\autoref{def:to_relative} to make inference techniques
for lower bounds of $\ircR$ applicable to $\pircR$.
The reason is that a lower
bound on $\ircR$ is not necessarily also a lower bound for $\pircR$
(the whole \emph{point} of performing innermost rewriting in parallel is to
reduce the asymptotic complexity!), so
using results by tools that compute lower bounds on $\ircR$ for
comparison would not make sense. We observe that non-trivial lower
bounds can be inferred for
155
out of the
186 examples proved confluent via \autoref{cor:confluence}.
This shows that
our transformation from \autoref{sec:dt_to_irc}
has practical value since it
produces relative TRSs
that are generally amenable to analysis by existing program analysis
tools.
Finally, \autoref{table:tight} shows that for overall
46 TRSs, the bounds that were found are asymptotically \emph{precise}.

\section{Related Work, Conclusion, and Future Work}
\label{sec:related}
\label{sec:conclusion}

\emph{Related work.}
We provide pointers to work on automated analysis of (sequential) innermost runtime complexity of
TRSs at the start of \autoref{sec:dt_to_irc}.
We now focus on
automated techniques
for
complexity analysis of
parallel/concurrent computation.

Our
notion of
parallel complexity follows a large
tradition of static \emph{cost analysis}, notably
for
concurrent programming. The two notable
works~\cite{AlbertLCTES11,AlbertTOCL18} address
async/finish programs where tasks are explicitly launched.
The authors propose several metrics such as the total number of spawned
tasks (in any execution of the program) and
a notion of
parallel complexity that is roughly the same as ours. They provide
static analyses that build on
techniques for estimating costs
of imperative languages with functions calls~\cite{Albert2012}, and/or
recurrence equations.
Recent approaches for the Pi
Calculus~\cite{BaillotPiESOP21,BaillotPiCONCUR21} compute
the
\emph{span} (our parallel complexity)
through a new typing system. Another type-based calculus for the same
purpose has been proposed with session types~\cite{HoffmannICFP18}.

For logic programs, which -- like TRSs -- express an implicit parallelism,
parallel complexity can be inferred
using recurrence solving~\cite{GallagherLOPSTR19}.

The tool \raml~\cite{raml} 
derives bounds on the worst-case evaluation cost of
first-order functional programs with list and pair
constructors as well as pattern matching and both sequential and
parallel composition \cite{HoffmannESOP15}.
They use two typing derivations with specially annotated types, one for the \emph{work}
and one for the \emph{depth} (parallel complexity).
Our
setting is more flexible wrt the shape of user-defined
data structures (we allow for tree constructors of arbitrary arity),
and our analysis
deals with both data structure and control in an integrated manner.

\emph{Conclusion and future work.}
We have defined parallel-innermost
runtime complexity for TRSs and proposed an approach to its
automated analysis.
Our approach allows for finding both upper and lower bounds and builds
on existing techniques and tools. Our experiments on the TPDB indicate that our
approach is practically usable, and we are
confident that it captures the potential parallelism of programs with pattern matching.

Parallel rewriting is a topic of active research, e.g., for GPU-based
massively parallel rewrite engines \cite{gpu_trs}. Here our work could
be useful to determine which functions to evaluate on the GPU.
More generally, parallelising
compilers which need to determine which function calls should be
compiled into parallel code may benefit from an analysis of
parallel-innermost runtime complexity such as ours. 

DTs have been used \cite{lctrsComplexity} in runtime complexity analysis of
\emph{Logically Constrained TRSs (LCTRSs)}~\cite{lctrs13}, an extension of
TRSs by built-in data types from SMT theories
(integers, arrays, \ldots).
This work could be extended to parallel rewriting.
Moreover, analysis of
\emph{derivational complexity} \cite{derivational89}
of parallel-innermost term rewriting can be a promising direction. Derivational complexity
considers the length of rewrite sequences from arbitrary start terms, e.g.,\linebreak
$\Fd(\Fd(\dots(\Fd(\FS(\FZero)))\dots))$ in \autoref{ex:doubles},
which can have longer derivations than basic terms of the same size.
Finally, towards
automated parallelisation
we aim to
infer
complexity bounds wrt\ term
\emph{height}
(terms = trees!), as suggested in~\cite{wst16trs}.

\medskip

\emph{Acknowledgements.} We thank the anonymous reviewers for helpful comments.

\newpage
\bibliography{references}

\newpage
\appendix
\section{Proofs}

\subsection{Proof of~\autoref{thm:cplxPar}}

To prove \autoref{thm:cplxPar}, we need some further definitions
and lemmas.

\begin{definition}[Argument Normal Form \cite{DependencyTuple}, Maximal Parallel Argument Normal Form]
A term $t$ is an \emph{argument normal form} iff $t \in \VV$ or
$t = f(t_1,\ldots,t_n)$ and all $t_i$ are in normal form.
A term $\maxanf{t}$ is a \emph{maximal parallel argument normal form} of a
term $t$ iff $\maxanf{t}$ is an argument normal form such that
$t \pitosbelow \maxanf{t}$ and for all argument normal forms
$t'$ with $t \pitosbelow t'$, we have $\Dh(t',\pito) \leq \Dh(\maxanf{t},\pito)$. Here $u \pitosbelow v$ denotes a rewrite sequence
with $\pito$ where all steps are at positions $> \varepsilon$.
\end{definition}

The following lemma is adapted to the parallel
setting from \cite{DependencyTuple}.

\begin{lemma}[Parallel Derivation Heights of Nested Subterms]
\label{lem:nested}
Let $t$ be a term, let $\RR$ be a TRS such that all
reductions of $t$ with $\ito$ are finite. Then
\begin{align*}
\Dh(t,\pito) &\leq
\max
\{
 \sum_{1 \leq i \leq k} \Dh(\maxanf{t|_{\pi_i}},\pito)
 \mid  \msdc{t}{\pi_1, \dots, \pi_k}
\}
\end{align*}

If $\pito$ is confluent, then we additionally have:
\begin{align*}
\Dh(t,\pito) &=
\max
\{
 \sum_{1 \leq i \leq k} \Dh(\maxanf{t|_{\pi_i}},\pito)
 \mid
 \msdc{t}{\pi_1, \dots, \pi_k}
\}
\end{align*}

\end{lemma}

\begin{proof}[of \autoref{lem:nested}]
By induction on the term size $|t|$.
If $|t| = 1$, the statement follows immediately
since $\maxanf{t}\ = t$.
Now consider the case $|t| > 1$. Let $n$ be the arity
of the root symbol of $t$. In (parallel-)innermost
rewriting, a rewrite step at the root of $t$ requires
that the arguments of $t$ have been rewritten to normal
forms. Since rewriting of arguments takes place in parallel
(case (b) of \defref{def:pito} applies), we have
\[
\Dh(t, \pito) \leq \Dh(\maxanf{t}, \pito) +
   \max \{ \quad \Dh(t|_j, \pito) \quad \mid 1 \leq j \leq n \}
\]
If $\pito$ is confluent, we additionally have equality
in the previous as well as in the next (in)equalities
since $\maxanf{t}$ is uniquely determined.

As $|t_j| < |t|$, we can apply the induction hypothesis:
\begin{align*}
\Dh(t, \pito) &\leq \Dh(\maxanf{t}, \pito) +
   \max \{\\
&    \hspace*{-4.5ex} \max \{
         \sum_{1 \leq i \leq m} \Dh(\maxanf{t|_{j.\tau_i}},\pito)
         \mid \msdc{t|_j}{\tau_1, \dots, \tau_m} \}
      \hspace*{2ex} \mid  1 \leq j \leq n \}
\end{align*}
Equivalently:
\begin{align*}
\Dh(t, \pito) &\leq
   \max \{  \Dh(\maxanf{t}, \pito) +
         \sum_{1 \leq i \leq m} \Dh(\maxanf{t|_{j.\tau_i}},\pito)
                   \mid 1 \leq j \leq n,\\[-2ex]
& \hspace*{42ex} \msdc{t|_j}{\tau_1, \dots, \tau_m}
                            \}\\
&=
\max
\{
 \sum_{1 \leq i \leq k} \Dh(\maxanf{t|_{\pi_i}},\pito)
 \mid \msdc{t}{\pi_1, \dots, \pi_k}
\}
\end{align*}

For the last equality, consider that
the maximal structural dependency chains $\pi_1, \dots, \pi_k$ of
$t$ can have two forms. If the root of $t$ is a defined symbol,
we have $\msdc{t}{j.\tau_1, \dots, j.\tau_m, \varepsilon}$.
Otherwise $\Dh(\maxanf{t}, \pito) = 0$ and thus
$\msdc{t}{j.\tau_1, \dots, j.\tau_m}$.
\qed
\end{proof}

We now can proceed with the proof of~\autoref{thm:cplxPar}.

\begin{proof}[of~\autoref{thm:cplxPar}]
As the first case, consider $\Dh(t, \pito) = \omega$.
Since $t$ is in argument normal form, the first rewrite step
from $t$ must occur at the root. Thus, there are
$\ell_1 \to r_1 \in \RR$ and a substitution $\sigma_1$
such that $t = \ell_1\sigma_1 \pito r_1\sigma_1$ and
$\Dh(r_1\sigma_1,\pito) = \omega$. Hence, there is a minimal subterm
$r_1\sigma_1|_{\pi_1}$ of $r_1\sigma_1$ such that
$\Dh(r_1\sigma_1|_{\pi_1},\pito) = \omega$ and all proper subterms of
$r_1\sigma_1|_{\pi_1}$ terminate wrt $\pito$. As $\sigma_1$
must instantiate all variables with normal forms, we have
$\pi_1 \in \PosDef(r_1)$, i.e., $r_1\sigma_1|_{\pi_1} =
r_1|_{\pi_1}\sigma_1$.
In the infinite $\pito$-reduction of $r_1|_{\pi_1}\sigma_1$,
all arguments are again reduced to normal forms first, and
we get a term $t'$ with $\Dh(t', \pito) = \omega$.
Since $t'$ is in argument normal form, the first rewrite step from
$t'$ must occur at the root. Thus, there are
$\ell_2 \to r_2 \in \RR$ and a substitution $\sigma_2$
such that $t = \ell_2\sigma_2 \pito r_2\sigma_2$ and
$\Dh(r_2\sigma_2,\pito) = \omega$.
This argument can be continued \emph{ad infinitum}, giving rise to an
infinite path in the chain tree
\[
(\tup{\ell}_1 \to \FCom_{n_1}(\ldots,\tup{r_1|}_{\pi_1},\ldots) \mid
\sigma_1), \qquad
(\tup{\ell}_2 \to \FCom_{n_2}(\ldots,\tup{r_2|}_{\pi_2},\ldots) \mid
\sigma_2), \qquad \dots
\]
Thus, $\tup{\ell}_1\sigma_1 = \tup{t}$ has an infinite chain
tree, and
$\Cplx_{\langle \DTpar(\RR), \DTpar(\RR), \RR \rangle}(\tup{t}) =
\omega$.

Now consider the case that $\Dh(t,\pito) \in \nat$.
We use induction on $\Dh(t,\pito)$. If $\Dh(t,\pito) = 0$,
the term $t$ is in normal form wrt $\RR$. Thus, $\tup{t}$ is in
normal form wrt $\DTpar(\RR) \cup \RR$, and
$\Cplx_{\langle \DTpar(\RR), \DTpar(\RR), \RR \rangle}(\tup{t}) = 0$.

If $\Dh(t,\pito) > 0$, since $t$ is in argument normal form, there are
$\ell \to r \in \RR$ and a substitution $\sigma$
such that $t = \ell\sigma \pito r\sigma = u$ and
\begin{align}
\Dh(t, \pito) &= 1 + \Dh(u, \pito) \label{eq:topstep}
\end{align}
As $\sigma$
must instantiate all variables with normal forms, we have
that $u|_\pi = r\sigma|_\pi$ is in normal form for all
$\pi \in \PosDef(u) \setminus \PosDef(r)$. For these positions $\pi$,
$\maxanf{u|_\pi} = u|_\pi$ and $\Dh(u|_\pi, \pito) = 0$.
From \lemref{lem:nested}, we get:
\begin{align}
&\quad\:\Dh(u,\pito) \notag\\
& \leq
\max
\{
 \sum_{1 \leq i \leq k} \hspace*{-1ex} \Dh(\maxanf{u|_{\pi_i}},\pito)
 \mid \msdc{u}{\pi_1, \dots, \pi_k}
\}\notag\\[1.5ex]
& =
\max
\{
 \sum_{1 \leq i \leq j} \hspace*{-1ex} \Dh(\maxanf{u|_{\pi_i}},\pito) +
 \hspace*{-2ex}
 \sum_{j+1 \leq i \leq k}
 \hspace*{-2ex}
 \Dh(\maxanf{u|_{\pi_i}},\pito)
 \mid \msdc{u}{\pi_1, \dots, \pi_k},
  \notag\\
& \hspace*{25ex}
\pi_1, \dots, \pi_j \in \PosDef(u) \setminus \PosDef(r),
\pi_{j+1}, \dots, \pi_k \in  \PosDef(r) \}\notag\\[1.5ex]
& =
\max
\{
 \sum_{1 \leq i \leq j} \hspace*{-1ex} \Dh(u|_{\pi_i},\pito) \,\:\; +
 \hspace*{-2ex}
 \sum_{j+1 \leq i \leq k}
 \hspace*{-2ex}
 \Dh(\maxanf{u|_{\pi_i}},\pito)
 \mid \msdc{u}{\pi_1, \dots, \pi_k},
 \notag\\
& \hspace*{25ex}
\pi_1, \dots, \pi_j \in \PosDef(u) \setminus \PosDef(r),
\pi_{j+1}, \dots, \pi_k \in  \PosDef(r) \}\notag\\[1.5ex]
& =
\max
\{ \hspace*{4ex}
 \sum_{j+1 \leq i \leq k} \Dh(\maxanf{u|_{\pi_i}},\pito)
 \mid \msdc{u}{\pi_1, \dots, \pi_k},
 \label{eq:maxanf}\\[-2ex]
& \hspace*{25ex}
\pi_1, \dots, \pi_j \in \PosDef(u) \setminus \PosDef(r),
\pi_{j+1}, \dots, \pi_k \in  \PosDef(r) \}\notag
\end{align}
Note that $\Dh(\maxanf{u|_{\pi}},\pito) \leq
\Dh(u|_{\pi},\pito) < \Dh(t,\pito)$ holds for all $\pi \in
\PosDef(r)$. Thus, with the induction hypothesis,
\eqref{eq:topstep} and \eqref{eq:maxanf}, we get:

\begin{align}
&\quad\:\Dh(t, \pito)\notag\\
& = 1 + \Dh(u, \pito)\notag\\
& \leq
1 +
\max
\{
 \sum_{j+1 \leq i \leq k}
  \hspace*{-1ex}
  \Cplx_{\langle \DTpar(\RR), \DTpar(\RR), \RR
                                       \rangle}(\tup{\maxanf{u|_{\pi_i}}})
\mid
 \msdc{u}{\pi_1, \dots, \pi_k},
 \label{eq:ind}\\[-2ex]
& \hspace*{18ex}
\pi_1, \dots, \pi_j \in \PosDef(u) \setminus \PosDef(r),
\pi_{j+1}, \dots, \pi_k \in  \PosDef(r) \}\notag
\end{align}
Let $\lmsdc \pi_1,\ldots,\pi_k \rmsdc$ be an arbitrary maximal structural
dependency chain for $r$. Then there exists a corresponding
chain tree for $\tup{t}$ whose root node is
$(\tup{\ell} \to \FCom_{k}(\tup{\maxanf{r_1|_{\pi_1}}},\ldots,
   \tup{\maxanf{r_k|_{\pi_k}}}) \mid \sigma)$
and where the children of the root node are maximal
chain trees for
$\tup{\maxanf{u|_{\pi_1}}},\ldots,\tup{\maxanf{u|_{\pi_k}}}$.
This follows because for all $1 \leq i \leq k$, we have
$r|_{\pi_i}\sigma = u|_{\pi_i}$ and so
$\tup{r|_{\pi_i}}\sigma \itos \tup{\maxanf{u|_{\pi_i}}}$.
Together with \eqref{eq:ind}, this gives
$\Dh(t,\pito) \leq
\Cplx_{\langle \DTpar(\RR), \DTpar(\RR), \RR \rangle}(\tup{t})$,
and for confluent $\pito$ we also get
$\Dh(t,\pito) =
\Cplx_{\langle \DTpar(\RR), \DTpar(\RR), \RR \rangle}(\tup{t})$.
\qed
\end{proof}

\subsection{Proof of~\autoref{thm:nopar}}

\begin{proof}[of \autoref{thm:nopar}]
Let $\RR$ be a TRS such that for all rules $\ell \to r \in \RR$,
$|\msdcSym(r)| = 1$.
\smallskip

We prove part (a).
By construction of $\DTpar$, for $\ell \to r \in \RR$ we get
from $|\msdcSym(r)| = 1$ that $|\DTpar(\ell \to r)| = 1$.
We show that for all rules $\ell \to r \in \RR$,
$\DTpar(\ell \to r) = \{ \DT(\ell \to r) \}$
and fix $\ell \to r \in \RR$.
$|\msdcSym(r)| = 1$ implies that $\PosDef(r)$ is ordered by
the prefix order $>$ on positions.
Thus, by using an arbitrary total extension of $>$
as the total order $\gtrdot$ used as an ingredient for the
construction of $\DT(\ell \to r)$, we obtain the result for
part (a).

\smallskip

We now prove part (b).
Let $t_0$ be a basic term for $\RR$ with a rewrite sequence
$t_0 \pito t_1 \pito t_2 \pito \dots$.
We show by induction over $i$ that for all $t_i$,
$t_i$ contains at most one innermost redex.

For the base case, consider that the basic term $t_0$
contains only a single occurrence of a defined symbol,
at the root.
Thus, if $t_0$ is a redex, it is also the unique innermost
redex in $t_0$.

For the induction step, assume that $t_i$ has at most one innermost redex.
If $t_i$ has no redex, it is a normal form, and we are done.
Otherwise, $t_i$ has exactly one innermost redex at position $\tau$,
and in the parallel-innermost rewrite step $t_i \pito t_{i+1}$
a rule $\ell \to r$ with matcher $\sigma$ replaces $t_i|_\tau = \sigma(\ell)$
by $\sigma(r)$. 
The premise $|\msdcSym(r)| = 1$ implies that there is exactly one
(empty or non-empty) maximal structural dependency chain
$\msdc{r}{\pi_1, \ldots, \pi_k}$.

Since the rewrite step $t_i \pito t_{i+1}$ uses
(parallel-)\emph{innermost} rewriting, $\sigma(x)$
is in normal form for all variables $x$.
Thus, potential redexes in term $t_{i+1}$ can only be at positions
$\tau \pi_1, \ldots, \tau \pi_k$.
As $\lmsdc \pi_1, \ldots, \pi_k \rmsdc$ is a structural
dependency chain, we have $\pi_1 > \dots > \pi_k$, which implies
$\tau \pi_1 > \dots > \tau \pi_k$. Thus, the term $t_{i+1}$ 
has at most one innermost redex $\tau \pi_i$.
This concludes part (b).
\smallskip

Part (c) follows directly from part (b) and the definitions of
$\pirc{\RR}(n)$ and $\irc{\RR}(n)$.

\end{proof}

\subsection{Proof of~\autoref{thm:upper_rel}}

For the proof of \autoref{thm:upper_rel}, we need the following
standard definition of a \emph{context}.

\begin{definition}[Context]
A \emph{context} $C[]$ is a term that contains exactly one occurrence of a
special symbol $\hole$. We write $C[t]$ for the term obtained from
replacing $\hole$ by the term $t$.
\end{definition}

\begin{proof}[of \autoref{thm:upper_rel}]
We first show part (a) of the statement.
For a DT Problem
$\langle \DD, \SSS, \RR \rangle$ and a term
$\tup{t} \in \SharpTerms$, consider an arbitrary
chain tree $T$. We will show that if $|T|_\SSS = n$, then also
$\Dh(\tup{t},\itodetup) \geq n$.
We consider two cases. First, $n = \omega$.
Since $T$ is finitely branching, there must be an infinite path
with infinitely many nodes
$
(\tup{u}_1 \to \FCom_{n_1}(\ldots,\tup{v_1},\ldots) \mid
\sigma_1),
(\tup{u}_2 \to \FCom_{n_2}(\ldots,\tup{v_2},\ldots) \mid
\sigma_2), \dots$ such that
$\tup{u}_1 \to \FCom_{n_1}(\ldots,\tup{v_1},\ldots),
\tup{u}_2 \to \FCom_{n_2}(\ldots,\tup{v_2},\ldots),
\ldots \in \DD$, for infinitely many $i_1 < i_2 < i_3 < \dots$,
we also have
$\tup{u}_i \to \FCom_{n_i}(\ldots,\tup{t_i},\ldots) \in \SSS$,
and for all $i$, we have
$\tup{v}_i\sigma_i \itos \tup{u}_{i+1}\sigma_{i+1}$.
Then we also have a corresponding infinite rewrite sequence
\begin{align*}
\tup{t} = \tup{u}_1\sigma_1 &\itosdetup
C_1[\tup{u}_{i_1}\sigma_{i_1}]
\someito{\SSS}
C_1[\FCom_{n_{i_1}}(\ldots,\tup{v}_{i_1},\ldots)\sigma_{i_1}]\\
&\someitos{(\DD\setminus\SSS)\cup\RR}\quad\;\;
C_2[\tup{u}_{i_2}\sigma_{i_2}]
\someito{\SSS}
C_2[\FCom_{n_{i_2}}(\ldots,\tup{v}_{i_2},\ldots)\sigma_{i_2}]\\
&\someitos{(\DD\setminus\SSS)\cup\RR}\quad\;\;
\dots
\end{align*}

\noindent
for some contexts $C_1,C_2,\ldots$ (which result from rewrite
steps with rules from $\DD$).

\medskip

Now consider the case $n \in \nat$. We use induction.
For $n = 0$, the statement trivially holds.
For the induction step, let $n > 0$.

The (potentially infinite)
chain tree $T$ has $m$ subtrees $T'_i$ with roots
$(\tup{u}_i \to \FCom_{q_i}(\tup{v}_{i,1}, \ldots, \tup{v}_{i,q_i})
\mid \sigma_i)$ such that
$\tup{u}_i \to \FCom_{q_i}(\tup{v}_{i,1}, \ldots, \tup{v}_{i,q_i}) \in
\SSS$ and the path in the chain tree from the root to
$(\tup{u}_i \to \FCom_{q_i}(\tup{v}_{i,1}, \ldots, \tup{v}_{i,q_i})
\mid \sigma_i)$
has no outgoing edges from a node with a DT in $\SSS$.

We show two separate statements in the induction step:

\begin{gather}
\text{For each $T_i$, the term $\tup{u}_i\sigma_i$ has
$\Dh(\tup{u}_i\sigma_i, \itodetup) \geq
|T'_i|_\SSS$.} \label{belowS}\\[1.5ex]
\begin{split}
\text{There are contexts $C_1,\ldots,C_m$ such that\qquad}\\[-3pt]
\tup{t} \itosdetup
\FCom_m(C_1[\tup{u_1}\sigma_1],\ldots,C_m[\tup{u_m}\sigma_m]).
\end{split}
\label{aboveS}
\end{gather}

Statements \eqref{belowS} and \eqref{aboveS} together imply
$\Dh(\tup{t}, \itodetup) \geq |T'_1|_\SSS + \dots + |T'_m|_\SSS  = n$.
\smallskip

On \eqref{belowS}:
Let $i \in \{1,\ldots,m\}$ be arbitrary and fixed, let
$u = u_i$, let $\sigma = \sigma_i$, let $T'=T'_i$ (to ease
notation).
So the root of $T'$ is
$(\tup{u} \to \FCom_{q}(\tup{v}_{1}, \ldots, \tup{v}_{q})
\mid \sigma)$.
Let this node have children
$N_1 = (\tup{w}_1 \to \FCom_{r_1}(\ldots) \mid \mu_1)$,
$\ldots$,
$N_q = (\tup{w}_q \to \FCom_{r_q}(\ldots) \mid \mu_q)$.
For the corresponding trees $T''_j$ with $N_j$ at the root,
we have $|T''_j|_\SSS < |T'|_\SSS \leq n$ by construction,
so the induction hypothesis is applicable to the terms
$\tup{w}_j\mu_j$, and we get
$\Dh(\tup{w}_j\mu_j, \itodetup) \geq |T''_j|_\SSS$ for all
$1 \leq j \leq q$. We construct a rewrite sequence with
$\itodetup$ using at least $1 + |T''_1|_\SSS + \dots +
|T''_q|_\SSS = |T'|_\SSS$ steps with a rule from $\SSS$ as follows:
\begin{align*}
\tup{u}\sigma &\someito{\SSS}
\FCom_q(\tup{v}_{1}\sigma, \ldots, \tup{v}_{q}\sigma)\\
& \someitos{\RR}
\FCom_q(\tup{w}_{1}\mu_1, \ldots, \tup{v}_{q}\sigma)\\
& \someitos{\RR}
\dots\\
& \someitos{\RR}
\FCom_q(\tup{w}_{1}\mu_1, \ldots, \tup{w}_{q}\mu_q)
\end{align*}

\noindent
With this rewrite sequence, we obtain \eqref{belowS} using the induction hypothesis:
\begin{align*}
&\quad\;\Dh(\tup{u}\sigma, \itodetup)\\
&\geq
1 + \Dh(\tup{w}_{1}\mu_1, \itodetup) + \dots +
\Dh(\tup{w}_{q}\mu_q, \itodetup)\\
&\geq 1 + |T''_1|_\SSS + \dots + |T''_q|_\SSS\\
&= |T'|_\SSS
\end{align*}

On \eqref{aboveS}:
Let the root of $T$ be
$(\tup{\ell} \to \FCom_p(\tup{r}_1,\ldots,\tup{r}_p) \mid \nu)$.
With a construction similar to the one used in the case
$n = \omega$, we get:
\begin{align*}
\tup{t} = \tup{\ell}\nu &\someito{\DD}\qquad\;\;\;
\FCom_p(\tup{r}_1\nu,\ldots,\tup{r}_p\nu)\\
&\someitos{(\DD\setminus\SSS)\cup\RR}
\FCom_p(C_1[\tup{u_1}\sigma_1],\ldots,\tup{r}_p\nu)\\
&\someitos{(\DD\setminus\SSS)\cup\RR}
\dots\\
&\someitos{(\DD\setminus\SSS)\cup\RR}
\FCom_p(C_1[\tup{u_1}\sigma_1],\ldots,C_m[\tup{u_m}\sigma_m])
\end{align*}

\noindent
for some contexts $C_1,\ldots,C_m$ (which result from rewrite
steps with rules from $\DD$). Note that here it suffices
to reduce only in those
subterms with a symbol $\tup{f}$ at their root
that are on a path to one of the
$C_i[\tup{u_i}\sigma_i]$, and depending on the tree structure,
each $\tup{r}_j\nu$ may yield 0 or more of these $m$ terms
(note that $p$ and $m$ are not necessarily equal).

This concludes the induction step and hence the overall proof
of part (a).
\smallskip

Part (b) follows from part (a), as shown in the following:

\begin{align*}
\irc{\langle \DD, \SSS, \RR \rangle}(n)
&= \sup \{ \Cplx_{\langle \DD, \SSS, \RR \rangle}(\tup{t}) \mid
   t \in \BasicTerms^\RR, \tsize{t} \leq n \} & \text{by \autoref{def:dt_problem}}\\
&\leq \sup \{ \Dh(\tup{t}, {\itodetup}) \mid
   t \in \BasicTerms^\RR, \tsize{t} \leq n \} & \text{by part (a)}\\
&\leq \sup \{ \Dh(s, {\itodetup}) \mid
   s \in \BasicTerms^{\RR\cup\DD}, \tsize{s} \leq n \}\\
&= \irc{\detupTRS}(n)
\end{align*}
\vspace*{-6ex}

\mbox{}\qed

\end{proof}

\subsection{Proof of~\autoref{thm:lower_rel}}

\begin{proof}[of \autoref{thm:lower_rel}]
We first consider the proof for part (a).

We use the following (many-sorted first-order
monomorphic) type assignment $\Theta$ with two sorts $\sbot$
and $\sDT$, where the arities of the symbols are respected (note that here all
arguments of a given symbol have the same type):
\begin{align*}
\Theta(f) & = \sbot \times \dots \times \sbot \to \sbot
              \text{ for $f$ in $\symsof{\RR}$}\\
\Theta(\tup{f}) & = \sbot \times \dots \times \sbot \to \sDT
                    \text{ for $\tup{f}$ a sharp symbol}\\
\Theta(\FCom_k) & = \sDT \times \dots \times \sDT \to \sDT
\end{align*}
With this type assignment, for all rules $\ell \to r \in \DD \cup
\RR$, $\ell$ and $r$ are well typed and have the same type: if $\ell
\to r \in \RR$, then all occurring symbols have the same result type
$\sbot$, which carries over to $\ell$ and $r$. And if $\ell \to r \in \DD$, then $\ell$ and $r$ have type
$\sDT$. To see that $\ell$ and $r$ are well typed, consider that every
term $\ell$ has the shape $\tup{f}(s_1,\ldots,s_n)$, where $\tup{f}$
has result type $\sDT$ and expects all arguments to have type $\sbot$,
while all $s_i$ contain only subterms of type $\sbot$. Similarly,
$r$ has the shape $\FCom_k(\tup{f}_1(t_{1,1},\ldots,t_{1,n_1}),
\ldots, \tup{f}_k(t_{k,1},\ldots,t_{k,n_k}))$.
$\FCom_k$ has result type $\sDT$ and expects all arguments to have
type $\sDT$. This is the case since all $\tup{f_i}$ have result type
$\sDT$. An all $\tup{f_i}$, which are right below the root, expect
their arguments $t_{i,j}$ to have result type $\sbot$. This is the
case by construction.

In the following, we consider basic terms that
are well typed according to $\Theta$ as start terms.
For our relative TRS $\SSS/((\DD\setminus\SSS)\cup\RR)$,
we have
the following
two kinds of well-typed basic terms that we need to consider:

\begin{description}
\item[Case 1:]
$t = f(t_1,\ldots,t_n)$ with $f \in \DefSyms$ and $t_1,\ldots,t_n
\in \TT(\ConSyms,\VV)$. This term and all its subterms are of type
$\sbot$. Thus, this term can be rewritten by rules from $\RR$, but not
by rules from $\DD$ (and $\SSS$), which all have type $\sDT$. As
rewriting preserves the type of terms, $t$ is a normal form wrt the relations
$\someito{\Theta(\SSS/((\DD\setminus\SSS)\cup\RR))}$ and
$\someito{\SSS/((\DD\setminus\SSS)\cup\RR)}$, and
$\Dh(t,\someito{\SSS/((\DD\setminus\SSS)\cup\RR)}) = 0$.
Since $\Cplx_{\langle \DD, \SSS, \RR \rangle}(s) \geq 0$
regardless of $s$, the claim follows for this case.

\medskip

\item[Case 2:]
$t = \tup{f}(t_1,\ldots,t_n)$ with $f \in \DefSyms$ and
$t_1,\ldots,t_n \in \TT(\ConSyms,\VV)$.
If $t$ is a normal form, there is no tree, and
$\Dh(t, \someito{\SSS/((\DD\setminus\SSS)\cup\RR)}) = 0 = \Cplx_{\langle \DD, \SSS, \RR \rangle}(t)$.

Otherwise, we can convert any ${\someito{\SSS \cup ((\DD\setminus\SSS)\cup\RR)}} =
{\someito{\DD \cup \RR}}$ rewrite
sequence to a $(\DD,\RR)$-chain tree $T$ for $t$ such that
$\Dh(t, \someito{\SSS/((\DD\setminus\SSS)\cup\RR)}) = |T|_\SSS$,
including any rewrite sequence that witnesses
$\Dh(t, \someito{\SSS/((\DD\setminus\SSS)\cup\RR)})$
in the following way:
\smallskip

As $t$ is a basic term, the first step in the rewrite sequence
rewrites at the root of the term. Since only rules from
$\DD$ are applicable to terms with $\tup{f}$ at the root,
this step uses a DT $\tup{s} \to \FCom_k(\ldots)$ from $\DD$.
With $\sigma$ as the used matcher for the rewrite step,
we obtain the root node $(\tup{s} \to \FCom_k(\ldots) \mid \sigma)$.
\smallskip

Now assume that we have a partially constructed chain tree $T'$ for the
rewrite sequence so far, which we have represented up until
the term $s$ that resulted from a $\someito{\DD}$ step.

If there are no further $\someito{\DD}$ steps in the rewrite sequence,
we have completed our chain tree $T = T'$ since the remaining $\someito{\RR}$
suffix of the rewrite sequence does not contribute
to $\Dh(t,\someito{\SSS/((\DD\setminus\SSS)\cup\RR)})$
(only steps using rules from $\SSS \subseteq \DD$ are counted).

Otherwise, our remaining rewrite sequence has the shape
$s \someitos{\RR} u \someito{\DD} v \someitom{\DD \cup \RR} \dots$
for some $m \in \nat \cup \{ \omega \}$.
The step $u \someito{\DD} v$ takes place at position $\pi$,
using the DT $\tup{p} \to \FCom_l(\tup{q}_1,
\ldots,\tup{q}_l) \in \DD$ and the matcher $\mu$.

We reorder the rewrite steps
$s \someitos{\RR} u$ by advancing all $\someito{\RR}$ steps
at positions $\tau > \pi$, yielding
$s \someitos{\RR,>\pi} s' \someitos{\RR,\not>\pi} u$.
Here $\someito{\RR,>\pi}$ denotes an innermost rewrite step
using rules from $\RR$ at a position $\tau > \pi$, and
$\someito{\RR,\not>\pi}$ denotes an innermost rewrite step
using rules from $\RR$ at a position $\tau' \not > \pi$.
Now we change our remaining rewrite sequence to
$s \someitos{\RR,>\pi} s' \someito{\DD} u' \someitos{\RR,\not>\pi} v \someitom{\DD \cup \RR} \dots$.
We encode the subsequence $s \someitos{\RR,>\pi} s' \someito{\DD} u'$
by adding the node $(\tup{p} \to \FCom_l(\tup{q}_1,
\ldots,\tup{q}_l) \mid \mu)$ to $T'$ as a child to
a node $N = (\tup{p'} \to \FCom_k(\tup{q'}_1,
\ldots,\tup{q'}_k) \mid \delta)$ where
$\FCom_k(\tup{q'}_1, \ldots,\tup{q'}_k) \delta = s|_\pi$
and where the subterm
$\FCom_k(\tup{q'}_1, \ldots,\tup{q'}_k)\delta$
in the DT of $N$
has not yet been used for this purpose in the
construction before.
That is, there is a $j$ such that $\tup{q'}_j \delta \someito{\RR} \tup{p}\mu$
that has not yet been used in the construction.
Such a node exists by construction.
We obtain the chain tree $T''$.

Now we can extend $T''$ further by encoding the rewrite sequence
$u' \someitos{\RR,\not>\pi} v \someitom{\DD \cup \RR} \dots$
following the same procedure.
Since our construction adds a node with a DT from $\SSS$ in
the first component of the label
whenever the rewrite sequence uses a rule from $\SSS$, we
have the desired property that
$\Dh(t, \someito{\SSS/((\DD\setminus\SSS)\cup\RR)}) = |T|_\SSS$.
This concludes the proof for part (a).

\bigskip

We now prove part (b).

Innermost runtime complexity is known to be a persistent
property wrt type introduction \cite{irc_persistent}.
For our relative TRS $\detupTRS$, this means that we may introduce
an arbitrary (many-sorted first-order monomorphic) type assignment
$\Theta$ for all symbols in the considered signature
such that the rules in $\RR\cup\DD$ are well typed.
We obtain a typed relative TRS $\Theta(\detupTRS)$, and
$\irc{\Theta(\detupTRS)}(n) = \irc{\detupTRS}(n)$
holds. Thus, only basic terms that are well typed according to
$\Theta$
need to be considered as start terms for $\irc{\detupTRS}$.
We write $\Theta(\BasicTerms^{\DD\cup\RR})$ for the set of well-typed
basic terms for $\detupTRS$.
\smallskip

We use the type assignment $\Theta$ from part (a)
to restrict the set of basic terms as start terms.
With this type assignment, we obtain:
\begin{align*}
&\phantom{{}={}}\irc{\detupTRS}(n)\\
&= \irc{\Theta(\detupTRS)}(n)
   & \text{by \cite{irc_persistent}}\\
&= \sup \{ \Dh(t, {\itodetup}) \mid
t \in \Theta(\BasicTerms^{\RR\cup\DD}), \tsize{t} \leq n \}
   & \text{by \autoref{def:rel}}\\
&\leq \sup \{ \Cplx_{\langle \DD, \SSS, \RR \rangle}(\tup{t}) \mid
t \in \Theta(\BasicTerms^{\RR\cup\DD}), \tsize{t} \leq n \}
   & \text{by part (a)}\\
&= \irc{\langle \DD, \SSS, \RR \rangle}(n)
\end{align*}

\vspace{-2ex}
\mbox{}\qed

\end{description}
\end{proof}

\end{document}

%% file: notations.tex
\newcommand{\confOrReport}[2]{#1}
\renewcommand{\confOrReport}[2]{#2}

\newcommand{\asym}{f}                

\newcommand{\sym}[1]{\mathsf{#1}}    
\newcommand{\FZero}{\sym{Zero}}
\newcommand{\FS}{\sym{S}}
\newcommand{\FNil}{\sym{Nil}}
\newcommand{\FTree}{\sym{Tree}}
\newcommand{\FCons}{\sym{Cons}}
\newcommand{\Fplus}{\sym{plus}}
\newcommand{\Fsize}{\sym{size}}
\newcommand{\Fdoubles}{\sym{doubles}}
\newcommand{\FCom}{\sym{Com}}
\newcommand{\FComPar}{\sym{ComPar}}

\newcommand{\Fa}{\sym{a}}
\newcommand{\Fb}{\sym{b}}
\newcommand{\Fc}{\sym{c}}
\newcommand{\Fd}{\sym{d}}

\newcommand{\Ff}{\sym{f}}
\newcommand{\Fg}{\sym{g}}

\newcommand{\Fminus}{\sym{-}}

\newcommand{\Fs}{\sym{s}}
\newcommand{\Fzero}{\sym{0}}

\newcommand{\Fleq}{\sym{leq}}
\newcommand{\Ftrue}{\sym{true}}
\newcommand{\Ffalse}{\sym{false}}
\newcommand{\Fmod}{\sym{mod}}
\newcommand{\Fif}{\sym{if}}

\newcommand{\tup}[1]{#1^{\sharp}}
\newcommand{\FPLUS}{\tup{\Fplus}}
\newcommand{\FSIZE}{\tup{\Fsize}}

\newcommand{\FDOUBLES}{\tup{\Fdoubles}}

\newcommand{\FD}{\tup{\Fd}}

\newcommand{\hole}{\Box}   
\newcommand{\sbot}{\alpha} 
\newcommand{\sDT}{\beta}  

\newcommand{\Dom}{\mathit{Dom}}
\newcommand{\Pol}{{\mathcal{P}ol}}

\newcommand{\Signature}{\ensuremath{\Sigma}}      
\newcommand{\DefSyms}{\ensuremath{\Signature_d}}      
\newcommand{\ConSyms}{\ensuremath{\Signature_c}}      
\newcommand{\symsof}[1]{\Signature^{#1}}
\DeclareMathOperator{\Dh}{dh}                     
\newcommand{\Cplx}{\mathit{Cplx}}                 
\newcommand{\rt}{\mathrm{root}}                   
\newcommand{\tsize}[1]{\ensuremath{\lvert#1\rvert}}

\newcommand{\irc}[1]{\mathrm{irc}_{#1}}
\newcommand{\ircR}{\irc{\RR}}
\newcommand{\pirc}[1]{\mathrm{irc}^{\parallel}_{#1}}
\renewcommand{\pirc}[1]{\mathrm{pirc}_{#1}}
\newcommand{\pircR}{\pirc{\RR}}

\newcommand{\maxanf}[1]{#1\!\Downarrow} 

\newcommand{\detup}{\delta} 

\newcommand{\detupTRS}{\SSS/((\DD\setminus\SSS)\cup\RR)}

\newcommand{\msdc}[2]{\msdcSym(#1;#2)}
\newcommand{\msdcSym}{\mathit{MSDC}}

\newcommand{\lmsdc}{\langle} 
\newcommand{\rmsdc}{\rangle} 

\renewcommand{\msdc}[2]{\lmsdc #2 \rmsdc \in \msdcSym(#1)}

\newcommand{\TT}{\mathcal{T}}
\newcommand{\RR}{\mathcal{R}}
\newcommand{\DD}{\mathcal{D}}
\newcommand{\SSS}{\mathcal{S}}

\newcommand{\RRA}{{\RR_1}}  
\newcommand{\RRB}{{\RR_2}}  

\newcommand{\OO}{\mathcal{O}}
\newcommand{\Oh}{\OO}
\newcommand{\VV}{\mathcal{V}}
\newcommand{\BasicTerms}{\TT_{\mathrm{basic}}}    
\newcommand{\SharpTerms}{\TT^\sharp}              

\newcommand{\Pos}{\mathcal{P}\!\mathit{os}}       
\newcommand{\PosDef}{\Pos_d}                      

\newcommand{\DT}{\mathit{DT}}                     
\newcommand{\DTpar}{\mathit{DT^{\parallel}}}      
\renewcommand{\DTpar}{\mathit{PDT}}

\newcommand{\ito}{
  \mathrel{\smash{\stackrel{\raisebox{2pt}{\scriptsize $\mathsf{i}\:$}}%
      {\smash{\rightarrow}}}_{\RR}}
}

\newcommand{\itodetup}{
  \mathrel{\smash{\stackrel{\raisebox{2pt}{\scriptsize $\mathsf{i}\:$}}%
      {\smash{\rightarrow}}}_{\detup(\langle \DD, \SSS, \RR \rangle)}}
}
\renewcommand{\itodetup}{
  \mathrel{\smash{\stackrel{\raisebox{2pt}{\scriptsize $\mathsf{i}\:$}}%
      {\smash{\rightarrow}}}_{\SSS/((\DD\setminus\SSS)\cup\RR)}}
}
\renewcommand{\itodetup}{\someito{\SSS/((\DD\setminus\SSS)\cup\RR)}}

\newcommand{\itosdetup}{
  \mathrel{\smash{\stackrel{\raisebox{2pt}{\scriptsize $\mathsf{i}\:$}}%
      {\smash{\rightarrow^*}}}_{\SSS/((\DD\setminus\SSS)\cup\RR)}}
}
\renewcommand{\itosdetup}{\someitos{\SSS/((\DD\setminus\SSS)\cup\RR)}}

\newcommand{\itop}{
  \mathrel{\smash{\stackrel{\raisebox{2pt}{\scriptsize $\mathsf{i}\:\:\:$}}%
      {\smash{\rightarrow^+}}}_{\RR}}
}
\renewcommand{\itop}{\someitop{\RR}}

\newcommand{\itos}{
  \mathrel{\smash{\stackrel{\raisebox{2pt}{\scriptsize $\mathsf{i}\:$}}%
      {\smash{\rightarrow^*}}}_{\RR}}
}
\renewcommand{\itos}{\someitos{\RR}}

\newcommand{\pito}{
  \mathrel{\smash{\stackrel{\raisebox{2pt}{\scriptsize $\:\:\mathsf{i}\:$}}%
      {\smash{\longrightarrow\hspace*{-16pt}{\parallel}\hspace*{8pt}}}}_{\RR}}
}

\newcommand{\pitos}{
  \mathrel{\smash{\stackrel{\raisebox{2pt}{\scriptsize $\:\:\mathsf{i}\:$}}%
      {\smash{\longrightarrow\hspace*{-16pt}{\parallel}\hspace*{8pt}}}}^*_{\RR}}
}

\newcommand{\pitosbelow}{
  \mathrel{\smash{\stackrel{\raisebox{2pt}{\scriptsize $\:\:\mathsf{i}$}}%
      {\smash{\longrightarrow\hspace*{-16pt}{\parallel}\hspace*{8pt}}}}^*_{\RR,>\varepsilon}}
}

\newcommand{\someto}[1]{
  \mathrel{\smash{\rightarrow}_{#1}}
}

\newcommand{\someito}[1]{
  \mathrel{\smash{\stackrel{\raisebox{2pt}{\scriptsize $\mathsf{i}\:$}}%
      {\smash{\rightarrow}}}_{#1}}
}
\newcommand{\someitos}[1]{
  \mathrel{\smash{\stackrel{\raisebox{2pt}{\scriptsize $\mathsf{i}\:$}}%
      {\smash{\rightarrow}}}_{#1}^{*}}
}
\newcommand{\someitop}[1]{
  \mathrel{\smash{\stackrel{\raisebox{2pt}{\scriptsize $\mathsf{i}\:$}}%
      {\smash{\rightarrow}}}_{#1}^{+}}
}
\newcommand{\someitom}[1]{
  \mathrel{\smash{\stackrel{\raisebox{2pt}{\scriptsize $\mathsf{i}\:$}}%
      {\smash{\rightarrow}}}_{#1}^{m}}
}

\newcommand{\aprove}{\textsc{AProVE}\xspace}
\newcommand{\tct}{\textsc{TcT}\xspace}

\newcommand{\starexec}{\textsc{StarExec}\xspace}
\newcommand{\cofloco}{\textsc{CoFloCo}\xspace}
\newcommand{\raml}{\textsc{RAML}\xspace}

\newcommand{\nat}{\mathbb{N}}

%% file: chain_tree.tikz.tex
\begin{tikzpicture}[align=left]
  \node (root) at (2,2) {
    $(\FSIZE(\FTree(v, l, r)) \to
      \FCom_3(\FSIZE(l), \FSIZE(r), \FPLUS(\Fsize(l), \Fsize(r)))$ \\
    \hspace*{48ex} $|~\{v \mapsto \FZero; l \mapsto \FNil; r \mapsto \FNil\})$
  };
  \node
  (right) at (1.4,0) {
    $(\FSIZE(\FNil) \to \FCom_0 \mid \{\})$
  };
  \node
  (left) at (-2.1,0.7) {
    $(\FSIZE(\FNil) \to \FCom_0 \mid \{\})$
  };
  \node
  (plus) at (5.2, -0.7) {
    $(\FPLUS(\FZero, y) \to \FCom_0 \mid \{y \mapsto \FZero\})$
  };
  \draw[->] (root.south) -- (left.north);
  \draw[->] (root.south) -- (right.north);
  \draw[->] (root.south) --
    node[above, sloped] {$\Fsize(\FNil) \itos \FZero$}
    (plus.north);
\end{tikzpicture}